\documentclass{article}
\usepackage[utf8]{inputenc}
\usepackage{a4wide}
\usepackage{xcolor}
\usepackage{amsmath}
\usepackage{lscape}
\usepackage[font=small,labelfont=bf]{caption}
\usepackage{graphicx}
\usepackage{subcaption}
\usepackage{comment}
\usepackage{tabularx}
\usepackage{multirow}
\usepackage{algpseudocode}
\usepackage{algorithm} 
\usepackage{array}
\usepackage{stfloats}
\usepackage{amsmath,amssymb,amsfonts,amsthm} 
{
\newtheorem{theorem}{Theorem}
\newtheorem{lemma}{Lemma}
\newtheorem{assumption}{Assumption}
\theoremstyle{remark}
\newtheorem {remark}{Remark}

}
\usepackage[utf8]{inputenc}
\usepackage{amssymb}
\usepackage{multicol}
\usepackage{geometry}
\usepackage{hyperref}
\usepackage{authblk}
\usepackage{ulem}
\usepackage{moreverb}
\usepackage{amsmath,bm}
\usepackage{graphicx}
\usepackage{siunitx}
\usepackage{subcaption}
\usepackage{cite}

\title{A Game Approach to Multi-dimensional Opinion Dynamics in Social Networks with Stubborn Strategist Agents} 

\author[1]{Hossein B. Jond}
\affil[1]{Department of Cybernetics, Faculty of Electrical Engineering, Czech Technical University in Prague, Czech Republic}
\author[2]{Aykut Y{\i}ld{\i}z}
\affil[2]{Department of Electrical and Electronics Engineering, TED University, Ankara, 06420, Türkiye}

\begin{document}

\maketitle

\begin{abstract}
In a social network, individuals express their opinions on several interdependent topics, and therefore the evolution of their opinions on these topics is also mutually dependent. In this work, we propose a differential game model for the multi-dimensional opinion formation of a social network whose population of agents interacts according to a communication graph. Each individual's opinion evolves according to an aggregation of disagreements between the agent's opinions and its graph neighbors on multiple interdependent topics exposed to an unknown extraneous disturbance. For a social network with strategist agents, the opinions evolve over time with respect to the minimization of a quadratic cost function that solely represents each individual's motives against the disturbance. We find the unique Nash/worst-case equilibrium solution for the proposed differential game model of coupled multi-dimensional opinions under an open-loop information structure. Moreover, we propose a distributed implementation of the Nash/worst-case equilibrium solution. We examine the non-distributed and proposed distributed open-loop Nash/worst-case strategies on a small social network with strategist agents in a two-dimensional opinion space. Then we compare the evolved opinions based on the Nash/worst-case strategy with the opinions corresponding to social optimality actions for non-strategist agents. 
 
\end{abstract}
\maketitle

\section{Introduction}\label{sec1}
Opinion dynamics in social networks has attracted a lot of attention in recent years from social sciences and control theory~\cite{FriedkinSociology,Acemoglu895377}. The extensive attention is due to the widespread applications in business, finance, marketing, e-commerce, politics, group decision-making, pandemic spread \cite{carballosa2021incorporating,you2022public}, and so on. See \cite{DONG201857,zha2020opinion} for surveys of applications of opinion dynamics. Researchers are looking for ideal mathematical models of interaction dynamics in a network of relationships to answer the critical question of how opinions form. Several models in this regard have been proposed to explain certain behaviors in the formation of opinions. The well-known ones are the French-DeGroot (FD) model \cite{degroot1974reaching}, the Friedkin-Johnsen (FJ) model \cite{friedkin1990social}, the Hegselmann-Krause (HK) model \cite{hegselmann2002opinion}, and Deffuant-Weisbuch (DW) model \cite{DeffuantWeisbuch}. The FD model of opinion dynamics addresses the emergence of consensus in a connected network topology. The FJ model includes stubborn agents and results in opinion diversity rather than consensus. In the HK model, the individuals are influenced by their nearest neighbors, so opinion clustering emerges rather than polarization or consensus. In the DW model, two randomly selected individuals take part in the opinion exchange business according to a threshold, resulting in consensus for a large selection of the threshold or clustering otherwise.

The network of interactions is usually defined as a graph to capture the patterns of interactions. In this context, the set of nodes represents the individuals, and each edge denotes a direct interaction, e.g., opinion transmission, between two individuals. On a communication graph with nonnegative weights assigned to its edges, the collaborative effort of all agents causes the network to reach an opinion consensus~\cite{CoDIT2023}. Networks with antagonistic agent interactions modeled by negative edge weights representing anticooperative behavior have also been studied~\cite{altafini,Su-Signed}. Multidimensional extensions of classical opinion dynamics models for multiple interdependent topics have recently been a topic of interest. The discrete-time multidimensional FJ and DG extensions are analyzed in~\cite{Parsegov} and their continuous-time counterparts with stubborn agents are presented in~\cite{YE2020108884}. A discrete-time multidimensional FJ model where each stubborn agent has a time-evolving stubbornness level is studied in~\cite{ZHOU2022170}.  

Previously, the classical opinion dynamics models considered social network agents as non-strategist entities, where each agent simply updated its opinions according to attractive interactions with its neighbors. However, when the agents of a social network become strategists, they become decision-makers who are selfish and prioritize their own interests in achieving the network's opinion consensus. Strategist agents with self-interests show noncooperative behavior, and thus opinion dynamics become a noncooperative multi-agent decision-making problem. In this problem, the decisions that other strategists are making at the same time impact each agent's opinion. For dynamical agents under the game theory framework, each agent, as a player of the game, determines her strategies in accordance with an individual optimization problem. Opinion dynamics in social networks has been studied using game theory in the literature for the last two decades~\cite{di2007opinion,GHADERI20143209,Abram39720,bauso2016opinion,bauso2018consensus}. An opinion formation game model for competition and bargaining among social network agents was studied in~\cite{Kareeva}. For selfish agents who take actions independently, the Nash equilibrium applies, and for agents who coordinate their behavior by negotiating their influence efforts, the Nash bargaining solution applies. Using a game setting, the cost of disagreement at best-response dynamics equilibrium was analyzed in~\cite{BINDEL2015248} by comparing the cost at equilibrium with the social optimum. The ratio of the worst equilibrium’s social cost to the optimal social cost as a measure of the price of anarchy in networks with directed graphs was studied for its boundedness in~\cite{CHEN2016808}. Public opinion evolution in the presence of conformity and manipulation behaviors via a game-theoretic approach was investigated in~\cite{Etesami2019}. Fast convergence with limited information exchange as a natural behavior for selfish agents under a game-theoretic model for social networks was studied in~\cite{Fotakis}. A game framework is used for optimal investment strategies for two competing camps, where the strategy of each camp comprises how much to invest in each node in a social network in~\cite{Dhamal2018,Dhamal2020,Dhamal2022}. The problem of competitive information spreading in the social network, e.g., see~\cite{Watkins}, as a zero-sum game that admits a unique Nash equilibrium in pure strategies is analyzed in~\cite{Mao2018,Mao2021}. Cost function learning, as opposed to the common assumption that the game players know each other’s cost function for decision-making, is analyzed for memorized social networks in~\cite{Mao2022}.

Recently, the differential game and its extension for large networks of agents, i.e., the mean-field game \cite{banez2022modeling}, were proven to provide models that are examined in a rigorous manner~\cite{Jiang212121,niazi2021differential,yildiz2021opinion,yildiz2022social,Paramahansa}. Differential games are one of the decentralized extensions of optimal control theory. Unlike in optimal control, there are multiple cost functions in a different game to minimize simultaneously \cite{bacsar1998dynamic}. Each agent tries to minimize her cost function, which consists of the deviation of her opinion from others (i.e., her graph neighbors) and her initial opinion as behavior characteristics expected from selfish agents.

In this paper, we propose a differential game model for multi-dimensional opinion formation in a social network on a communication graph. The proposed model is a multi-dimensional extension to differential game models in~\cite{niazi2021differential,yildiz2021opinion,yildiz2022social} and additionally, an unknown disturbance is included in the formulation. Furthermore, the proposed model is a differential game contribution to the current literature in multi-dimensional coupled opinion dynamics studies~\cite{Parsegov,Mengbin,Friedkin2,Noorazar,Ye}. The agents communicate on several unrelated or interdependent topics on an information graph. The game relies on the assumption that the players know each other’s cost functions for decision-making. Under the noncooperative mode of play, we derive an explicit open-loop solution for the Nash/worst-case equilibrium and its associated opinion formation trajectory. Additionally, we propose a distributed implementation of the solution based on the information available to each agent. The proposed game provides an analytical insight into the opinion dynamics of social networks via Nash equilibrium. Consequently, the main contributions of this work are summarized as follows: First, a differential model for multi-dimensional opinion dynamics on directed graphs is established. The cost functions include a stubbornness term for each agent. Second, the game strategies for each agent were found in an explicit expression derived by applying the necessary conditions for optimality using Pontryagin’s principle. Third, we propose a distributed implementation of game strategies that each agent can execute only using information from her graph neighbors.

The rest of the manuscript is organized as follows. In Section \ref{section:preliminaries}, some basic concepts in graph theory and the linear quadratic differential game are presented. In Section \ref{section:game}, the game of opinion formation is introduced. In Section \ref{section:open-loop-Nash}, the explicit solution of the open-loop Nash/worst-case equilibrium is derived, and in Section \ref{section:distributed}, its distributed implementation is discussed. Opinion formation as social optimality is considered in Section \ref{section:optimal}. In Section \ref{section:simulation}, numerical simulation results are demonstrated on a small social network with five agents. The state transition matrix is found in explicit form by the spectral decomposition for the commutative linear time-varying systems discussed in~\ref{app:time-varying}. Finally, Section \ref{section:conclusions} is dedicated to the conclusions and future works. 

\section{Preliminaries and Notation}\label{section:preliminaries}
The differential game framework presented in this paper for opinion formation in networks on graphs relies on some basic concepts in graph theory and the linear quadratic differential game that are discussed in the following. 
\subsection{Graph theory}
A directed graph is a pair $\mathcal{G}(\mathcal{V},\mathcal{E})$ where $\mathcal{V}$ is a finite set of vertices or nodes and $\mathcal{E}\subseteq \{(i,j):i,j \in \mathcal{V}\}$ is a set of edges or arcs. Each edge $(i,j)\in \mathcal{E}$ represents an information flow from node $i$ to node $j$ and is assigned a positive weight $\omega_{ij}>0$. The pair $(i,i)$ represents a self-loop. The set of neighbors of vertex $i$ is defined by $\mathcal{N}_i=\{j \in \mathcal{V}:(i,j)\mbox{ or }(j,i)\in \mathcal{E},j\neq i\}$. Graph $ \mathcal{G}$ is connected if for every pair of vertices $(i,j)\in \mathcal{V}\times\mathcal{V}$, from $ i$ to $ j$ for all $ j  \in  \mathcal{V}, j \neq  i $, there exists a path of (undirected) edges from $\mathcal{E}$. Matrix $D\in \mathbb{R}^{\mathcal{\lvert V \rvert }\times\mathcal{\lvert E \rvert }}$ is the incidence matrix of $\mathcal{G}$ where $D$'s $uv$th element is 1 if the node $u$ is the head of the edge $v$, $-1$ if the node $u$ is the tail, and $0$, otherwise.

An undirected graph is a pair $\mathcal{G}(\mathcal{V},\mathcal{E})$ where $\mathcal{V}$ is a finite set of vertices and $\mathcal{E}\subseteq \{\{i,j\}:i,j \in \mathcal{V}\}$ is a set of edges. Each edge $\{i,j\}$ represents bidirectional information flow between node $i$ and node $j$. The set of neighbors of vertex $i$ is defined by $\mathcal{N}_i=\{j \in \mathcal{V}:\{i,j\}\in \mathcal{E},j\neq i\}$. For an undirected graph, when constructing $D$, the direction of the edges is arbitrary.

The graph Laplacian matrix is defined as
\begin{equation}\label{eq:Laplacian}
    L=DWD^\top,
\end{equation}
where $W=\mathrm{diag}(\cdots,\omega_{ij},\cdots) \in \mathit{\mathbb{R}}^{\mathcal{\lvert E \rvert}\times \mathcal{\lvert E \rvert}}$ $\forall (i,j) \in \mathcal{E}$. In other words, the $ij$th entry of the Laplacian matrix is -$\omega_{ij}$, and its $i$th diagonal entry is $\sum_{j\in \mathcal{N}_i} \omega_{ij}$. The Laplacian $L$ is symmetric $(L=L^\top)$, positive semidefinite (PSD) $(L\geq 0)$, and satisfies the sum-of-squares property \cite{Olfati-Saber}
\begin{equation} \label{eq:sum-of-squares}
\sum_{\{i,j\}\in \mathcal{E}} \omega_{ij} \| y_i-y_j \|^2 =y^\top Ly,   
\end{equation}
where $ y=[y_1,\cdots,y_n]^\top$ is a nonzero vector and $\| .\|$ is the Euclidean norm. The nonzero eigenvalues of Laplacian $L$ have nonnegative real parts. A more compact notation for the matrix expression on the right hand of (\ref{eq:sum-of-squares}) is $\Vert y\Vert_{L}^2$. For any square matrix $A$, the compact form $\Vert y \Vert_{A}^2$  means $\Vert y \Vert_{A}^2=y^\top A y$.

The Kronecker product extends the dimension of a matrix \cite{Laub}. The extended dimensional incidence matrix $\mathcal{D}$ is defined as 
\begin{equation*}
\mathcal{D}=D_{n\times \mathcal{\lvert E \rvert}} \otimes I_{m\times m}=
\begin{bmatrix}
d_{11}I_{m\times m} & \cdots & d_{1\mathcal{\lvert E \rvert}}I_{m\times m}\\
\vdots & \ddots & \vdots\\
d_{n1}I_{m\times m} & \cdots & d_{n\mathcal{\lvert E \rvert}}I_{m\times m}
\end{bmatrix}
\in \mathit{\mathbb{R}}^{nm \times \mathcal{\lvert E \rvert}m},         
\end{equation*}
where $D=[d_{ij}]$, $I_{m\times m}$ is the identity matrix of dimension $m$, and $\otimes$ demonstrates the Kronecker product operator. According to the Kronecker product properties, $(D\otimes I_{m\times m})^\top=(D^\top\otimes I_{m\times m})=\mathcal{D}^\top$. 

The Laplacian matrix definition in (\ref{eq:Laplacian}) can be revisited as
\begin{equation*}
    \mathcal{L}=(D\otimes I_{m\times m})\mathcal{W}(D\otimes I_{m\times m})^\top=\mathcal{D}\mathcal{W}\mathcal{D}^\top,
\end{equation*}
where $\mathcal{L}$ is the extended dimensional Laplacian matrix, $\mathcal{W}=\mathrm{diag}(\cdots,\omega_{ij},\\\cdots)~\forall \{i,j\} \in \mathcal{E}$ is a block diagonal weighting matrix and $\omega_{ij}\in \mathit{\mathbb{R}}^{m\times m}$. 

\subsection{Linear Quadratic Nash/Worst-Case Equilibria} \label{sec:Nash}

Consider an $N$-player noncooperative linear-quadratic differential game with an unknown disturbance described by the following linear differential equation~\cite{Engwerda2017}
\begin{equation}\label{eq:dynamics-dist}
    \dot{x}(t)=A(t)x(t)+\sum_{i=1}^NB_i(t)u_i(t)+\sum_{i=1}^NB_{\varpi_i}(t)\varpi_i(t), \quad x(0)=x_0
\end{equation}
where $A(t)\in\mathit{\mathbb{R}}^{n\times n}$, $B_i(t)\in\mathit{\mathbb{R}}^{n\times m_i}$, $B_{\varpi_i}(t)\in\mathit{\mathbb{R}}^{n\times k_i}$, $x(t)\in\mathit{\mathbb{R}}^{n\times 1}$ is the state vector, $u_i(t)\in\mathit{\mathbb{R}}^{m_i\times 1}$ denotes the input control vector for player $i$, and $\varpi_i(t)\in\mathit{\mathbb{R}}^{k_i\times 1}$ is an extraneous input representing the unknown disturbance. We assume that this unknown disturbance is finite in the sense that $\int_0^\infty \Vert\varpi_i(t)\Vert^2~\mathrm{dt}$ exists as a finite number (i.e. $\varpi_i(t)$ is square integrable or $\varpi_i(t)\in L_2^\dagger$. 

Each player has a quadratic cost function
\begin{equation}\label{eq:cost-dist}
   J_i=\frac{1}{2}x^\top(t_f)Q_ix(t_f)+\frac{1}{2} \int_{0}^{t_f} \big(u_i^\top(t)R_iu_i(t)-\varpi_i^\top(t)R_{\varpi_i}\varpi_i(t)\big)~\mathrm{dt}, 
\end{equation}
where $Q_i\in\mathit{\mathbb{R}}^{n\times n}$, $R_i\in\mathit{\mathbb{R}}^{m_i\times m_i}$ and $R_{\varpi_i}\in\mathit{\mathbb{R}}^{k_i\times k_i}$ are positive definite for $i=1,\cdots,N$.

Due to the uncooperativeness assumption, a so-called set of Nash equilibrium actions is expected to be played. By adopting a Nash/worst-case equilibrium, every player has no incentive to change her policy given her worst-case expectations of the disturbance and the actions of other players. Under open-loop information, a player’s policy will depend on information only from the beginning of the game.

Let $\hat{E}_{Nn}$ denote the block-column matrix of $N$ blocks of the identity matrix of dimension $n$ and $\bar{E}_{i,j}$ be the block-column matrix of $i$ blocks of zero matrices of dimension $n$ except for block $j$ which is an identity matrix. $\mathrm{diag}(A)_N$ also denotes the $N\times N$ block diagonal matrix with diagonal blocks $A$.    

As mentioned in \textit{Corollary 3.2} in~\cite{Engwerda2017}, the linear-quadratic differential game (\ref{eq:dynamics-dist}) and (\ref{eq:cost-dist}) is rewritten as the $2n$-player differential game
\begin{equation}\label{eq:dynamics-dist-2n}
    \dot{\bar{x}}(t)=\bar{A}_N(t)\bar{x}(t)+\sum_{i=1}^N\bar{B}_i(t)u_i(t)+\sum_{i=1}^N\bar{B}_{\varpi_i}(t)\varpi_i(t), \quad \bar{x}(0)=\bar{x}_0,
\end{equation}
where player $i$ tries to minimize with respect to $u_i(t)$ and player $N+i$ tries to maximize with respect to $\varpi_i(t)$ the cost function
\begin{equation}\label{eq:cost-dist-2n}
   \bar{J}_i=\frac{1}{2}\bar{x}^\top(t_f)\bar{Q}_i\bar{x}(t_f)+\frac{1}{2} \int_{0}^{t_f} \big(u_i^\top(t)R_iu_i(t)-\varpi_i^\top(t)R_{\varpi_i}\varpi_i(t)\big)~\mathrm{dt}, 
\end{equation}
where $\bar{x}(t)=\hat{E}_{Nn}x(t)$, $\bar{B}_i(t)=\hat{E}_{Nn}{B}_i(t)$, $\bar{B}_{\varpi_i}(t)=\bar{E}_{N,i}{B}_{\varpi_i}(t)$, and $\bar{Q}_i=\bar{E}_{N,i}Q_i\bar{E}_{N,i}^\top$.

Nash/worst-case equilibria for the linear-quadratic differential game (\ref{eq:dynamics-dist-2n}) and (\ref{eq:cost-dist-2n}) are derived using Pontryagin’s principle in terms of a nonsymmetric Riccati equation. 

Let $G=\mathrm{diag}(R_1,\cdots,R_n)$, which is invertible since $R_i$ ($i=1,\cdots,N$) are, $B(t)=[\bar{B}_1(t),\cdots,\bar{B}_n(t)]^\top$, $\Tilde{B}^\top(t)=\mathrm{diag}(B_1^\top(t),\cdots,B_n^\top(t))$, $\Tilde{S}(t)=BG^{-1}\Tilde{B}^\top(t)-\mathrm{diag}(S_{\varpi_i})$, $\Tilde{M}(t)=\begin{bmatrix}
    \bar{A}_N(t) & -\Tilde{S}(t) \\ -Q & -\bar{A}_N^\top(t)
\end{bmatrix}$, $\Tilde{S}(t)=\hat{E}_{Nn}[S_1(t),\cdots,\\S_N(t)]-\mathrm{diag}(S_{\varpi_i})$, $Q=\mathrm{diag}(Q_i)$, $S_j(t)=B_j^\top(t)R_j^{-1}B_j(t)$, and $S_{\varpi_i}(t)=B_{\varpi_i}^\top(t)R_{\varpi_i}^{-1}B_{\varpi_i}(t)$.

In~\cite{Engwerda2017}, it is shown that if the nonsymmetric Riccati differential equation
\begin{align}\label{eq:Riccati-2n}
&\dot{P}(t)+\bar{A}_N^\top(t)P(t)+P(t)\bar{A}_N(t)-P(t)\Tilde{S}(t)P(t)=0,\quad P(t_f)=Q,
\end{align}
has a solution on $[0,t_f]$, the linear-quadratic differential game (\ref{eq:dynamics-dist}) and (\ref{eq:cost-dist}) has a unique Nash/worst-case equilibrium for every initial state $x_0$.
The worst-case Nash equilibrium actions and corresponding worst-case disturbances are given by the following
\begin{align*}
    u(t)=-G^{-1}\Tilde{B}^\top(t)P(t)\bar{x}(t), \\
    \varpi_i(t)=R_{\varpi_i}^{-1}B_{\varpi_i}^\top(t)P_i(t)\bar{x}(t),
\end{align*}
where $u(t)=[u_1^\top(t),\cdots,u_N^\top(t)]^\top$, $P^\top(t)=[P_1^\top(t),\cdots,P_N^\top(t)]$, $P_i(t)\in\mathit{\mathbb{R}}^{n\times nN}$. Moreover, the state trajectory satisfies
\begin{equation*}
    \dot{\bar{x}}(t)=A_{cl}(t)\bar{x}(t), \quad \bar{x}(0)=\bar{x}_0,
\end{equation*}
where
\begin{equation*}
   A_{cl}(t)=\bar{A}_N(t)-BG^{-1}\Tilde{B}^\top(t)P(t)\bar{x}(t)+\mathrm{diag}(S_{\varpi_i})P(t).
\end{equation*}

\section{Game of Opinion Formation}\label{section:game}

Consider a social network of $n$ (heterogeneous) agents indexed by $1$ through $n$ on a communication graph $\mathcal{G}(\mathcal{V},\mathcal{E})$ communicating about $m$ interdependent topics ($m\in \mathit{\mathbb{N}}>1$). The set of vertices $\mathcal{V}=\{1,\cdots,n\}$ corresponds to the set of agents. Each edge $(i,j)\in \mathcal{E}$ represents an opinion flow from node $i$ to node $j$ and is assigned an interpersonal influences matrix $W_{ij}\in \mathit{\mathbb{R}}^{m\times m}$ with nonnegative entries. Without a loss of generality, assume there are two interdependent topics (i.e., $m=2$), namely, topic $p$ and topic $q$. The off-diagonal entry $\omega_{ij}^{pq}>0$ in $W_{ij}=\begin{bmatrix}\omega_{ij}^{pp} & \omega_{ij}^{pq} \\ \omega_{ij}^{qp} & \omega_{ij}^{qq}\end{bmatrix}$ imposes a coupling between topic $p$ of agent $i$ and topic $q$ of agent $j$. Similarly, $\omega_{ij}^{qp}>0$ relates to a coupling between topic $q$ of agent $i$ and topic $p$ of agent $j$. Their magnitude corresponds to the level of trust accorded one agent in the opinion of the other agent. $W_{ii}=\begin{bmatrix}\omega_{ii}^{pp} & 0 \\ 0 & \omega_{ii}^{qq}\end{bmatrix}$ is called the stubbornness matrix for agent $i$, where its entries are the stubbornness coefficients. If agent $i$ is not stubborn about a topic, the corresponding diagonal entry for that topic becomes zero. Similarly, the off-diagonal entries for topics that are not evolving in coupling by the opinion dynamics model are also set to zero.

\begin{assumption} \label{assumption:connectivity}
 The opinion network graph $\mathcal{G}$ for every topic is connected.
\end{assumption}

The connectivity of $\mathcal{G}$ means there exists at least one globally reachable node for every topic (a root node of a spanning tree on the graph). In an opinion network on a connected communication graph, each agent for every topic has at least one neighbor with whom they interact. Two agents $i$ and $j$ are said to be graph neighbors if a topic in agent $i$ is coupled to another topic or the same topic in another agent $j$ through a nonzero matrix $W_{ij}$.

Let $x_i(t)\in \mathit{\mathbb{R}}^{m\times 1}$ denote the opinion vector of agent $i\in\mathcal{V}$ at time $t>0$ which evolves in continuous time according to the following dynamics~\cite{Ahn}  
\begin{equation} \label{eq:dynamics00}
 \dot{x}_i(t)=\sum_{j\in\mathcal{N}_i}W_{ij}(x_j(t)-x_i(t)).   
\end{equation}
The meaning of this model is that each agent’s opinion on multiple interdependent topics evolves according to the sum of disagreements between the opinions of each agent and its graph neighbors on those topics.

Traditionally, an agent in a social graph is modeled as a node that simply aggregates the opinions of other agents to update its own opinion according to a specified rule without strategically influencing other agents' opinions. In this study, we assume that every agent in the social network graph $\mathcal{G}(\mathcal{V},\mathcal{E})$ strategically and directly influences the opinions of its neighboring agents on the graph. Therefore, all agent interactions, restricted by the graph topology, are strategic. The agents decide on their influence efforts, i.e., find strategies under a noncooperative game framework, by minimizing a personal cost function that consists of opinion disagreement costs with their neighborhood and their own initial opinions (i.e., stubbornness costs). 

For optimization purposes for strategizing agents in a social network, let the vector $u_i(t)\in \mathit{\mathbb{R}}^{d\times 1}$ denote the agent's control or influence, and the input $\varpi_i(t)$ be an extraneous input representing an unknown disturbance influencing the system ($\varpi_i(t)\in L_2$). Then, the dynamics (\ref{eq:dynamics00}) becomes
\begin{equation} \label{eq:dynamics0}
 \dot{x}_i(t)=\sum_{j\in\mathcal{N}_i}W_{ij}(x_j(t)-x_i(t))+b_iu_i(t)+\xi_i\varpi_i(t).   
\end{equation}
Here, $b_i\in \mathit{\mathbb{R}}^{m\times d}$ and $\xi_i\in \mathit{\mathbb{R}}^{m\times d}$ represent the dependency structure of multiple-control inputs. $b_i$'s entries denote the level of trust of agent $i\in\mathcal{V}$ in the declared opinion or her actions (influence efforts). $\xi_i$ models the disturbance's internal structure from the point of view of agent $i\in\mathcal{V}$.

Let us define the following finite time horizon quadratic cost for agent $i\in \mathcal{V}$ to minimize in $t\in[0,t_f]$ in order to determine her actions 
\begin{align}
\label{eq:quadratic-cost0}
J_i=\frac{1}{2}\Vert x_i(t_f) - x_i(0)\Vert_{W_{ii}}^2&+\frac{1}{2} \sum_{j\in \mathcal{N}_i}\Vert x_i(t_f) - x_j(t_f)\Vert_{W_{ij}}^2+\nonumber\\&\frac{1}{2} \int_{0}^{t_f}\big(\Vert u_i(t) \Vert_{R_i}^2-\Vert\varpi_i(t)\Vert_{R_{\varpi_i}}^2\big)~\mathrm{dt}, 
\end{align}
where $R_i\in \mathit{\mathbb{R}}^{d\times d}$ and $R_{\varpi_i}\in \mathit{\mathbb{R}}^{d\times d}$ are symmetric positive definite, and $t_f$ is the terminal time. Matrix $R_{\varpi_i}$ models the agent's expectation about the severity of the disturbance $\xi_i\varpi_i(t)$. Agents who expect insignificant disturbances choose large $R_{\varpi_i}$~\cite{Engwerda2017}. A finite terminal time refers to situations where the attitudes of social groups are analyzed in a limited amount of time, like during a political or advertising campaign with an end date. 

The cost function in (\ref{eq:quadratic-cost0}) has three terms. The first term is a weighted difference between the terminal opinions vector $x_i(t_f)$ and the initial opinions vector $x_i(0)$ (or the agent's prejudices). Agents with prejudices (i.e., $x_i(0)\neq 0$) are stubborn, thus this term reflects the agent's stubbornness. The second term is a weighted sum of differences between the terminal opinions vector $x_i(t_f)$ and the terminal opinions vector $x_j(t_f)$ for each neighboring $j\in\mathcal{N}_i$ according to the communication graph $\mathcal{G}$. The last term is the weighted control or influence effort made during the entire opinion formation process, considering his expectations of the disturbance.

The proposed cost function in this paper unifies the graph definition of a network with a control effort in a multi-dimensional opinion space exposed to disturbances. By adopting it, each agent is trying to minimize their disagreement with their initial opinions and with the opinions of their neighbors at the terminal time with the least influence effort during the whole process of opinion formation. By moving toward their neighbors’ opinions, players can prevent additional costs from being incurred by their neighbors~\cite{BINDEL2015248}. Therefore, opinion formation behavior is characterized by minimizing the cost function at equilibrium. To that end, we like to find a control policy that minimizes the cost of disagreement with the least effort against a worst-case assumption of the disturbance. 

By concatenating the states of all agents into a vector $x(t)=[x_1^\top(t),\cdots,\\x_n^\top(t)]^\top$, the state equation with a given initial state is given by
\begin{equation} \label{eq:dynamics}
 \dot{x}(t)=-\mathcal{L}x(t)+\sum_{i=1}^nB_iu_i(t)+\sum_{i=1}^nB_{\varpi_i}\varpi_i(t), \quad x(0)=x_0, 
\end{equation}
where $\mathcal{L}=\mathcal{D}\mathcal{W}\mathcal{D}^\top$, and $\mathcal{W}=\mathrm{diag}(\cdots,W_{ij},\cdots)~\forall (i,j) \in \mathcal{E}$, $B_i=[0,\cdots,b_i^\top,\\\cdots,0]^\top$, $B_{\varpi_i}=[0,\cdots,\xi_i^\top,\cdots,0]^\top$. 

According to the sum-of-squares property of graph Laplacian, $J_i$ can be expressed as  
\begin{align}
\label{eq:quadratic-cost}
J_i&=\frac{1}{2} \Vert x(t_f) - x_0\Vert_{F_{i}}^2 +\frac{1}{2}\Vert x(t_f)\Vert_{\mathcal{L}_i}^2 +\frac{1}{2} \int_{0}^{t_f} \big(\Vert u_i(t) \Vert_{R_i}^2-\Vert\varpi_i(t)\Vert_{R_{\varpi_i}}^2\big)~\mathrm{dt}, 
\end{align}
where $F_i=\mathrm{diag}(0,\cdots,W_{ii},\cdots,0)$, $\mathcal{L}_i=\mathcal{D}\mathcal{W}_i\mathcal{D}^\top$, and $\mathcal{W}_i=\mathrm{diag}(0,\cdots,W_{ij}\\,\cdots,0)~\forall j \in \mathcal{N}_i$.

The cost function (\ref{eq:quadratic-cost}) and the dynamics (\ref{eq:dynamics}) pose a differential game problem \cite{Engwerda}. 
This category of problems is very promising for modeling and analyzing conflict situations in networked systems of self-interested dynamical agents. In the context of a differential game, each agent of the network is referred to as a player and thereby the game of opinion formation in (\ref{eq:quadratic-cost}) and (\ref{eq:dynamics}) has $n$ players. In this context, each player seeks the control $u_i(t)$ that minimizes her cost function $J_i$ subject to the game state equation with the given initial state in (\ref{eq:dynamics}). In other words, the players in the game seek to minimize their cost functions in order to find their control or influence strategies $u_i(t)$ while their opinions evolve according to the differential equation (\ref{eq:dynamics}).

The behavior of self-interested players in the game of opinion formation is best reflected via noncooperative game theory. Under the framework of noncooperative games, the players can not make binding agreements, and therefore, the solution (i.e., the Nash equilibrium) has to be self-enforcing, meaning that once it is agreed upon, nobody has the incentive to deviate from \cite{Damme}. In the next section, we derive the open-loop Nash/worst-case equilibrium solution for the game of coupled multi-dimensional opinion formation in terms of (\ref{eq:quadratic-cost}) and (\ref{eq:dynamics}).

\section{Nash/Worst-case Equilibrium Solution}\label{section:open-loop-Nash}

A Nash equilibrium is a strategy combination of all players in the game with the property that no one can gain a lower cost by unilaterally deviating from it. The open-loop Nash equilibrium is defined as a set of admissible actions ($u_1^*,\cdots,u_n^* $) if for all admissible ($u_1,\cdots,u_n$) the inequalities
$J_i (u_1^*,\cdots,u_{i-1}^*,u_i^*,u_{i+1}^*,\cdots,u_n^* )\leq
J_i (u_1^*,\cdots,u_{i-1}^*,u_i,u_{i+1}^*,\cdots,u_n^* )$
hold for $i\in\{1,\cdots,n\}$ where $u_i\in\Gamma_i$ and $\Gamma_i=\{u_i(t,x_0)|t\in[0,t_f]\}$ is the admissible strategy set for player $i$. 
The noncooperative linear quadratic differential game and the unique Nash equilibrium associated with it are discussed thoroughly in \cite{Engwerda}. 

Assume that for every strategy combination ($u_1,\cdots,u_n$) there exists a worst-case disturbance $\hat{\varpi}_i(u_1,\cdots,u_n)$ from the point of view of player $i$, i.e., $J_i\big(u_1,\cdots,u_n,\hat{\varpi}_i(u_1,\cdots,u_n)\big)\geq J_i\big(u_1,\cdots,u_n,\varpi\big)$ holds for each $\varpi$ and $i\in\{1,\cdots,n\}$. Then, a Nash/worst-case equilibrium is defined as a set of admissible actions ($u_1^*,\cdots,u_n^*$) that the inequalities
$J_i \big(u_1^*,\cdots,u_n^*,\hat{\varpi}_i(u_1^*,\cdots,u_n^*)\big)\\\leq
J_i \big(u_1^*,\cdots,u_{i-1}^*,u_i,u_{i+1}^*,\cdots,u_n^*,\hat{\varpi}_i(u_1^*,\cdots,u_{i-1}^*,u_i,u_{i+1}^*,\cdots,u_n^*)\big)$ 
hold for $i\in\{1,\cdots,n\}$ where $u_i\in\Gamma_i$. 

As presented in Subsection~\ref{sec:Nash}, the linear quadratic differential game (\ref{eq:dynamics}) and (\ref{eq:quadratic-cost}) converts to the nonsymmetric Riccati differential equation problem (\ref{eq:Riccati-2n}). Since solving (\ref{eq:Riccati-2n}) analytically is difficult, its approximate solution has to be obtained by numerical methods. Thereby, the Nash/worst-case equilibrium actions $u_i(t)$ and their corresponding worst-case disturbances $\varpi_i(t)$ in Subsection~\ref{sec:Nash} could only be approximated. 

In the following, we present the open-loop Nash/worst-case equilibrium solution for the underlying game of opinion formation. The solution is an explicit expression derived by applying the necessary conditions for optimality using Pontryagin’s principle. Before we begin, we present the following assumptions and lemmas.

\begin{lemma}\label{lemma-matrix}
Let $X^+$ and $\mathcal{R}(X)$ denote the Moore-Penrose inverse and the range, respectively, of a real matrix $X$. For nonnegative definite (or positive semidefinite) $X\in \mathit{\mathbb{R}}^{n\times n}$ and $Y\in \mathit{\mathbb{R}}^{n\times n}$, $X \overset{L}{\preceq}Y$ means $X$ is below $Y$ with respect to the Loewner partial ordering. $X \overset{L}{\preceq}Y$ holds if and only if $\mathcal{R}(X)\subseteq\mathcal{R}(Y)$ and $\mu_{\max}(Y^+X)\leq1$, where $\mu_{\max}(.)$ denotes the maximum eigenvalue of a matrix. 
\end{lemma}
The proof can be found in~\cite{BAKSALARY1991135}. $X \overset{L}{\preceq}Y$ is equivalent to $Y-X$ being nonnegative definite for two nonnegative definite matrices $X$ and $Y$~\cite{GRO1997457}.
\begin{assumption}\label{assumption:L}
    $B_iR_i^{-1}B_i^\top \overset{L}{\preceq}B_{\varpi_i}R_{\varpi_i}^{-1}B_{\varpi_i}^\top$.
\end{assumption}

\begin{lemma}\label{lemma-nonsingular} 
 Let $\Phi(t_f,t)$ denote the state-transition matrix for $-\mathcal{L}$ (and since $-\mathcal{L}$ is time-invariant, $\Phi(t_f,t)=\mathrm{e}^{-(t_f-t)\mathcal{L}}$). Matrix
\begin{equation}\label{eq:Htf}
    H(t_f)=I-\sum_{i=1}^n\big(\Psi_i(t_f,0)(F_i+\mathcal{L}_i)\big),
\end{equation}
where 
 \begin{equation}\label{eq:matrix-Psi-i}
\Psi_i(t_f,0)=\int_0^{t_f}\Phi(t_f,\tau)\big(B_{\varpi_i}R_{\varpi_i}^{-1}B_{\varpi_i}^\top-B_iR_i^{-1}B_i^\top\big)\Phi^\top(t_f,\tau)~\mathrm{d\tau}, 
 \end{equation}
is nonsingular if Assumption~\ref{assumption:L} holds. 
\end{lemma}
\begin{proof}
Firstly, the eigenvalues of $F_i+\mathcal{L}_i$ have nonnegative real parts. This is because \textit{i}) the nonzero eigenvalues of Laplacian $L_i$ and $\mathcal{L}_i$ have nonnegative real parts, and \textit{ii}) the eigenvalues of nonnegative diagonal matrix $W_{ii}$ and diagonal $F_i$ are real nonnegative. Secondly, the eigenvalues of $\Psi_i(t_f,0)$ defined in (\ref{eq:matrix-Psi-i}) are real and nonnegative. Let us rewrite (\ref{eq:matrix-Psi-i}) as $\Psi_i(t_f,0)=\Psi_{i1}(t_f,0)-\Psi_{i2}(t_f,0)$ where
 \begin{align*}
 \Psi_{i1}(t_f,0)&=\int_0^{t_f}\Phi(t_f,\tau)B_{\varpi_i}R_{\varpi_i}^{-1}B_{\varpi_i}^\top\Phi^\top(t_f,\tau)~\mathrm{d\tau},\\
 \Psi_{i2}(t_f,0)&=\int_0^{t_f}\Phi(t_f,\tau)B_iR_i^{-1}B_i^\top\Phi^\top(t_f,\tau)~\mathrm{d\tau}.
 \end{align*}
Both $\Psi_{i1}(t_f,0)$ and $\Psi_{i2}(t_f,0)$ are positive definite (see Theorem 12.6.18. in~\cite{Matrix-book}) as well as symmetric (see Fact 3.7.2.in~\cite{Matrix-book}). According to Assumption~\ref{assumption:L}, $\Psi_{i}(t_f,0)$ is positive semidefinite. Thus, its eigenvalues are real and nonnegative. According to its spectrum property, the eigenvalues of the product $\Psi_i(t_f,0)(F_i+\mathcal{L}_i)$ are the eigenvalue products between $\Psi_i(t_f,0)$ and $F_i+\mathcal{L}_i$ and hence have nonnegative real parts. Finally, the eigenvalues of $H(t_f)$ have positive real parts, and the existence of $H^{-1}(t_f)$ is concluded.
\end{proof}

\begin{remark}
    For the case when Nash/worst-case equilibrium reduces to Nash equilibrium, in (\ref{eq:matrix-Psi-i}) $B_{\varpi_i}R_{\varpi_i}^{-1}B_{\varpi_i}^\top=0$, and the nonsingularity of $H(t_f)$ is obvious from Lemma~\ref{lemma-nonsingular} without Assumption~\ref{assumption:L}.  
\end{remark}

\begin{theorem}\label{theorem:open-Nash}
 The opinion formation as the $n$-player noncooperative differential game in (\ref{eq:dynamics}) and (\ref{eq:quadratic-cost}) admits a unique open-loop Nash/worst-case equilibrium. The unique Nash/worst-case actions, corresponding worst-case disturbances, as well as the associated state trajectory are given by
 \begin{align}
         u_i(t&)=-R_i^{-1}B_i^\top\Phi^\top(t_f,t)\Big[(F_i+\mathcal{L}_i)H^{-1}(t_f)G(t_f)-F_i\Big]x_0, \label{eq:open-Nash}\\
         \varpi_i(t&)=R_{\varpi_i}^{-1}B_{\varpi_i}^\top\Phi^\top(t_f,t)\Big[(F_i+\mathcal{L}_i)H^{-1}(t_f)G(t_f)-F_i\Big]x_0, \label{eq:dist}\\
         x(t)&=\Big(\Phi(t,0)+\Psi(t,0)\Big[(F+P)H^{-1}(t_f)G(t_f)-F\Big]\Big)x_0,  \label{eq:open-Nash-tr}
 \end{align}
 where
 \begin{align}
     &H(t_f)=I-\Psi(t_f,0)(F+P), \label{eq:matrix-H}\\ &G(t_f)=\Phi(t_f,0)-\Psi(t_f,0)F,\label{eq:matrix-G}\\
  &\Psi(t_f,0)=[\Psi_1(t_f,0),\cdots,\Psi_n(t_f,0)],\label{eq:matrix-Psi}\\
  &F=[F_1^\top,\cdots,F_n^\top]^\top, \label{eq:matrix-F}\\
  &P=[\mathcal{L}_1^\top,\cdots,\mathcal{L}_n^\top]^\top.\label{eq:matrix-P}
 \end{align}
\end{theorem}

\begin{proof}
The Hamiltonian for every $i\in\mathcal{V}$ is defined as 
\begin{align*}
        \mathcal{H}_i=\frac{1}{2}\Vert u_i(t) \Vert_{R_i}^2&-\frac{1}{2}\Vert\varpi(t)\Vert_{R_{\varpi_i}}^2+\nonumber\\&\lambda_i^\top(t) \Big(-\mathcal{L}x(t)+\sum_{i=1}^nB_iu_i(t)+\sum_{i=1}^nB_{\varpi_i}\varpi_i(t)\Big), 
\end{align*}
where $\lambda_i(t)$ is a co-state vector. According to Pontryagin’s principle, the necessary conditions for optimality are
\begin{align*}
    &\frac{\partial \mathcal{H}_i}{\partial u_i}=0, \quad \frac{\partial \mathcal{H}_i}{\partial \varpi_i}=0, \quad \dot{\lambda}_i(t)=-\frac{\partial \mathcal{H}_i}{\partial x}. 
\end{align*}
Equivalently, 
\begin{align}
    &u_i(t)=-R_i^{-1}B_i^\top\lambda_i(t), \label{eq:neccesary-u}\\
     &\varpi_i(t)=R_{\varpi_i}^{-1}B_{\varpi_i}^\top\lambda_i(t), \label{eq:neccesary-d}\\
    &\dot{\lambda}_i(t)=\mathcal{L} \lambda_i(t), \label{eq:neccesary-lam}
\end{align}
subject to the terminal condition
\begin{equation}\label{eq:Lambda-compact}
        \lambda_i(t_f)=(F_i+\mathcal{L}_i)x(t_f)-F_ix_0.
\end{equation}

The solution of (\ref{eq:neccesary-lam}) is 
\begin{equation}\label{eq:LTVsoln}
\lambda_i(t)=\Phi(t_f,t)\lambda_i(t_f).
\end{equation} 
Substituting this solution in (\ref{eq:neccesary-u}) and (\ref{eq:neccesary-d}) yield
\begin{equation}\label{eq:Nash00}
   u_i(t)=-R_i^{-1}B_i^\top\Phi(t_f,t)\lambda_i(t_f), 
\end{equation}
\begin{equation}\label{eq:Nash00d}
   \varpi_i(t)=R_{\varpi_i}^{-1}B_{\varpi_i}^\top\Phi(t_f,t)\lambda_i(t_f). 
\end{equation}

Substituting (\ref{eq:Nash00}) and (\ref{eq:Nash00d}) in the state dynamics equation, we get 
\begin{equation}\label{eq:dynamics-lambda}
    \dot{x}(t)=-\mathcal{L}x(t)+\sum_{i=1}^n\big(B_{\varpi_i}R_{\varpi_i}^{-1}B_{\varpi_i}^\top-B_iR_i^{-1}B_i^\top\big)\Phi(t_f,t)\lambda_i(t_f).
\end{equation}
Equation (\ref{eq:dynamics-lambda}) has a solution at $t_f$ as follows
\begin{align}\label{eq:state-equation-tf}
x(t_f)=\Phi(t_f,0)x_0+\Psi(t_f,0)\lambda(t_f),
\end{align}
with $\Psi(t_f,0)$ given in (\ref{eq:matrix-Psi}). Stacking (\ref{eq:Lambda-compact}) for $i = 1, \cdots, n$ yields
\begin{equation} \label{eq:noncompact}
    \begin{bmatrix}
        \lambda_1(t_f)\\ \vdots \\ \lambda_n(t_f)
    \end{bmatrix}=\begin{bmatrix}
        (F_1+\mathcal{L}_1)x(t_f)\\ \vdots \\ (F_n+\mathcal{L}_n)x(t_f)
    \end{bmatrix}-\begin{bmatrix}
        F_1x_0\\ \vdots \\ F_nx_0
    \end{bmatrix}.
\end{equation} For the sake of notational simplicity, (\ref{eq:noncompact}) is shown as
\begin{equation}\label{eq:lamtf-compact}
    \lambda(t_f)=(F+P)x(t_f)-Fx_0,
\end{equation}
with $F$ and $L$ defined in (\ref{eq:matrix-F}) and (\ref{eq:matrix-P}), respectively, and $\lambda(t_f)=[\lambda_1^\top(t_f),\cdots,\\\lambda_n^\top(t_f)]^\top$. Substituting $\lambda(t_f)$ from (\ref{eq:lamtf-compact}) then into $x(t_f)$ in (\ref{eq:state-equation-tf}) yields 
 \begin{align*}
x(t_f)=\Phi(t_f,0)x_0+\Psi(t_f,0)(F+P)x(t_f)-\Psi(t_f,0)Fx_0,
\end{align*}
or alternatively,
 \begin{align}\label{eq:state-equation-simplify}
\Big(I-\Psi(t_f,0)(F+P)\Big)x(t_f)=\Big(\Phi(t_f,0)-\Psi(t_f,0)F\Big)x_0.
\end{align}

Therefore, the game has an open-loop Nash equilibrium for any initial state $x_0$ if and only if (\ref{eq:state-equation-simplify}) is satisfied for any arbitrary final state $x(t_f)$. If so, the equilibrium actions are unique and exist for all $t\in[0,t_f]$. Otherwise, the game does not have a unique open-loop Nash equilibrium for every initial state $x_0$. 

Using the notation $H(t_f)$ and $G(t_f)$ in (\ref{eq:matrix-H}) and (\ref{eq:matrix-G}), respectively, we can rewrite (\ref{eq:state-equation-simplify}) as
 \begin{align}\label{eq:state-x(tf)}
x(t_f)=H^{-1}(t_f)G(t_f)x_0.
\end{align}

Finally, by substituting (\ref{eq:state-x(tf)}) into (\ref{eq:Lambda-compact}) and re-substituting (\ref{eq:Lambda-compact}) in (\ref{eq:neccesary-u}) and (\ref{eq:neccesary-d}), we obtain (\ref{eq:open-Nash}) and (\ref{eq:dist}), respectively. 

The solution of (\ref{eq:dynamics-lambda}) at $t$ is given by
\begin{equation*}
    x(t)=\Phi(t,0)x_0+\Psi(t,0)\lambda(t_f). 
\end{equation*}
Substituting $\lambda(t_f)$ from (\ref{eq:lamtf-compact}) and then re-submitting $x(t_f)$ from (\ref{eq:state-x(tf)}), we have 
\begin{align}
x(t)&=\Phi(t,0)x_0+\Psi(t,0)(F+P)x(t_f)-\Psi(t,0)Fx_0 \nonumber \\
&=\Phi(t,0)x_0+\Psi(t,0)(F+P)H^{-1}(t_f)G(t_f)x_0-\Psi(t,0)Fx_0, \nonumber 
\end{align}
or in its final form (\ref{eq:open-Nash-tr}). 
This concludes the proof.
\end{proof}

From Theorem~\ref{theorem:open-Nash}, it is obvious that the unique Nash/worst-case equilibrium actions and their corresponding worst-case disturbances exist if and only if $H(t_f)$ in (\ref{eq:Htf}) has an inverse. 

Let $M=(F_i+\mathcal{L}_i)H^{-1}(t_f)G(t_f)-F_i$. So,
\begin{align*}
    u_i(t)=-R_i^{-1}B_i^\top\Phi^\top(t_f,t)Mx_0, \\
    \varpi_i(t)=R_{\varpi_i}^{-1}B_{\varpi_i}^\top\Phi^\top(t_f,t)Mx_0.
\end{align*}
Additionally, assume that $R_i=R_i^\top$ and $R_{\varpi_i}=R_{\varpi_i}^\top$. Thus, $(R_i^{-1})^\top=(R_i^\top)^{-1}=R_i^{-1}$.

The worst-case expected costs by player $i$ are
\small
\begin{align*}
    &J_i=\frac{1}{2} \int_{0}^{t_f} \big( u_i^\top(t)R_iu_i(t)-\varpi_i^\top(t)R_{\varpi_i}\varpi_i(t)\big)~\mathrm{dt}\nonumber\\&=-\frac{1}{2}x_0^\top M^\top \int_{0}^{t_f} \Phi(t_f,t)\big( B_{\varpi_i}R_{\varpi_i}^{-1}B_{\varpi_i}^\top-B_iR_i^{-1}B_i^\top\big)\Phi^\top(t_f,t)~\mathrm{dt}Mx_0\nonumber\\&=-\frac{1}{2}x_0^\top M^\top \Psi_i(t_f,0)Mx_0.
\end{align*}
\normalsize
The terminal opinion formation error for player $i$ is
\begin{align}\label{eq:terminal-error-i}
E_i&=\frac{1}{2}\Vert x_i(t_f) - x_i(0)\Vert_{W_{ii}}^2+\frac{1}{2} \sum_{j\in \mathcal{N}_i}\Vert x_i(t_f) - x_j(t_f)\Vert_{W_{ij}}^2\nonumber\\&=\frac{1}{2} \Vert x(t_f) - x_0\Vert_{F_{i}}^2 +\frac{1}{2}\Vert x(t_f)\Vert_{\mathcal{L}_i}^2.
\end{align}

\section{Distributed Implementation of Game Strategies}\label{section:distributed}

It is obvious from (\ref{eq:open-Nash}), (\ref{eq:dist}), and (\ref{eq:open-Nash-tr}) that the Nash/worst-case actions, their corresponding worst-case disturbances, and state trajectory require the vector $x_0$ from all players or nodes in the communication graph. Consequently, control actions (\ref{eq:open-Nash}) and (\ref{eq:dist}) are compatible only with a fully connected communication graph where each agent has access to global information. A complete graph resembles a social network in which each agent knows all other agents and communicates with them. In practice, since each agent usually does not know all the other agents, the graph is incomplete. In such a network, each player has access to her immediate neighbors' information in the opinion network's communication graph and does not share her information with non-neighbor agents. Therefore, the distributed implementation of control actions (\ref{eq:open-Nash}) and (\ref{eq:dist}) is necessary for incomplete graphs. Moreover, distributed control policies that could be implemented only with information from neighboring nodes have several advantages, such as reliability, scalability, and flexibility~\cite{distributed}. However, (\ref{eq:open-Nash}) and (\ref{eq:dist}) are incapable of being executed with the distributed information.

In the following, we propose a distributed implementation of (\ref{eq:open-Nash}) and (\ref{eq:dist}), which each agent can execute only with information from her graph neighbors. Inspired by the Nash strategy design approach in a distributed manner for UAV (unmanned aerial vehicle) formation control in~\cite{Lin2014}, for each agent $i$, we propose the following state estimator
 \begin{align}\label{eq:state-est}
         \hat{\dot x}_i(t)=&\delta_i\bigg\{\big(e_i^\top\Phi(t_f,0)e_i-e_i^\top\Psi_i(t_f,0)e_iW_{ii}\big)x_i(0)\nonumber\\&-\hat{x}_i(t)+e_i^\top\Psi_i(t_f,0)e_i\sum_{j\in\mathcal{N}_i}W_{ij}(\hat{x}_j(t)-\hat{x}_i(t))\bigg\},
 \end{align}
where $\delta_i$ is a positive scalar and $e_i=[0_{m\times m},\cdots,I_{m\times m},\cdots,0_{m\times m}]^\top$.

\begin{theorem}\label{theorem:distributed}
If each agent $i$ for all $i\in\mathcal{V}$ updates its state estimate $\hat{x}_i(t)$ according to (\ref{eq:state-est}), then
\begin{equation}\label{eq:estimate}
    \lim_{t\to\infty} \hat{x}(t)=x(t_f)
\end{equation}
where $\hat{x}(t)=[\hat{x}_1^\top(t),\cdots,\hat{x}_n^\top(t)]^\top$.
\end{theorem}
\begin{proof}
Stacking (\ref{eq:state-est}) from $1$ to $n$ yields
\begin{align}\label{eq:linear-system}
         \hat{\dot x}(t)&=\Delta\left(\big(\Phi(t_f,0)-\Psi(t_f,0)F\big)x_0-\Big(I-\sum_{i=1}^n\big(\Psi_i(t_f,0)(F_i+\mathcal{L}_i)\big)\Big)\hat{x}(t)\right)\nonumber\\&=\Delta\left(\big(\Phi(t_f,0)-\Psi(t_f,0)F\big)x_0-(I-\Psi(t_f,0)(F+P))\hat{x}(t)\right)
         \nonumber\\&=\Delta\left(\big(\Phi(t_f,0)-\Psi(t_f,0)F\big)x_0-H(t_f)\hat{x}(t)\right) 
 \end{align}
where $\Delta=\mathrm{diag}(g_1,\cdots,g_n)\otimes I_{m\times m}$. 

Using the notation $H(t_f)$, (\ref{eq:state-equation-simplify}) can be rewritten as
 \begin{align}\label{eq:state-equation-simplify-0}
H(t_f)x(t_f)=\Big(\Phi(t_f,0)-\Psi(t_f,0)F\Big)x_0.
\end{align}
Substituting the left side from (\ref{eq:state-equation-simplify-0}) into (\ref{eq:linear-system}), we have
\begin{align}\label{eq:linear-system-f}
         \hat{\dot x}(t)=-\Delta H(t_f)\left(\hat{x}(t)-x(t_f)\right) 
 \end{align}
As shown in Lemma~\ref{lemma-nonsingular}, all the eigenvalues of matrix $H(t_f)$ have positive real parts, and thus matrix $-\Delta H(t_f)$ is Hurwitz. Therefore, the linear system in (\ref{eq:linear-system-f}) is asymptotically stable, and $\hat{x}(t)$ will converge to $x(t_f)$ at $t\to\infty$.
\end{proof}

Using the distributed state estimator (\ref{eq:state-est}), each agent $i\in\mathcal{V}$ can adopt the following distributed strategy
 \begin{align}\label{eq:open-Nash-distributed}
         \hat{u}_i(t&)=-R_i^{-1}B_i^\top\Big(W_{ii}(\hat{x}_i(t)-x_i(0))+\sum_{j\in\mathcal{N}_i}W_{ij}(\hat{x}_j(t)-\hat{x}_i(t))\Big), 
 \end{align}
and its corresponding distributed worst-case disturbance is
 \begin{align}\label{eq:dist-distributed}
         \hat{\varpi}_i(t&)=R_{\varpi_i}^{-1}B_{\varpi_i}^\top\Big(W_{ii}(\hat{x}_i(t)-x_i(0))+\sum_{j\in\mathcal{N}_i}W_{ij}(\hat{x}_j(t)-\hat{x}_i(t))\Big). 
 \end{align}

All agents will eventually converge to the open-loop Nash/worst-case strategy, its corresponding disturbance, and state trajectory with the arbitrary convergence speed determined by $\delta_i$ if they adopt distributed implementation (\ref{eq:open-Nash-distributed}), (\ref{eq:dist-distributed}), and (\ref{eq:state-est}). Notice that each agent requires only her own initial state and state estimates and the state estimates of her graph neighbors, which can be acquired at the beginning of the game.

The terminal opinion formation error for player $i$ when adopting distributed strategies is
\begin{align}\label{eq:terminal-error-i-dist}
\hat{E}_i&=\frac{1}{2}\Vert \hat{x}_i(t_f) - x_i(0)\Vert_{W_{ii}}^2+\frac{1}{2} \sum_{j\in \mathcal{N}_i}\Vert \hat{x}_i(t_f) - \hat{x}_j(t_f)\Vert_{W_{ij}}^2\nonumber\\&=\frac{1}{2} \Vert \hat{x}(t_f) - x_0\Vert_{F_{i}}^2 +\frac{1}{2}\Vert \hat{x}(t_f)\Vert_{\mathcal{L}_i}^2.
\end{align}

\section{Social Optimality}\label{section:optimal}

In the game of multi-dimensional opinion formation on the communication graph $\mathcal{G}(\mathcal{V},\mathcal{E})$, each player by minimizing the individual (local) cost function attains the locally optimal Nash equilibrium. However, the game equilibrium in general does not correspond to the global social optimum, which minimizes the sum of all players' costs. 
The global network dynamics can be rewritten as
\begin{equation} \label{eq:dynamics-g}
 \dot{x}(t)=-\mathcal{L}x(t)+Bu(t)+B_{\varpi}\varpi(t), \quad x(0)=x_0, 
\end{equation}
where $B=[B_1,\cdots,B_n]$, $u(t)=[u_1^\top(t),\cdots,u_n^\top(t)]^\top$, $B_{\varpi}=[B_{\varpi_1},\cdots,B_{\varpi_n}]$, and $\varpi(t)=[\varpi_1^\top(t),\cdots,\varpi_n^\top(t)]^\top$.
Social optimality utilizes an identical global objective (e.g., a cost function to minimize) as the opinion formation objective by simply adding all the individual players’ interests together. Therefore, the social optimality problem is associated with a social network with non-strategist agents in which individuals do not wish to execute their personal strategies but rather prefer to commit to the whole network's common objective. The global cost function becomes  
\begin{equation*}
    J=\sum_{i=1}^nJ_i,
\end{equation*}
or equivalently
\begin{align}
\label{eq:quadratic-cost-g}
J&=\frac{1}{2} \Vert x(t_f) - x_0\Vert_{\bar{F}}^2 +\frac{1}{2}\Vert x(t_f)\Vert_{\mathcal{L}}^2+\frac{1}{2} \int_{0}^{t_f}\big(\Vert u(t) \Vert_R^2-\Vert \varpi(t) \Vert_{R_{\varpi}}^2\big)~\mathrm{dt} ,
\end{align}
where $\bar{F}=\mathrm{diag}(W_{11},\cdots,W_{nn})$, $R=\mathrm{diag}(R_1,\cdots,R_n)$, and $R_{\varpi}=\mathrm{diag}(R_{\varpi_1},\\\cdots,R_{\varpi_n})$. 
The global cost function in (\ref{eq:quadratic-cost-g}) is known to all $i\in\mathcal{V}$ and is minimized with respect to (\ref{eq:dynamics-g}).

\begin{theorem}\label{theorem:optimal}
The socially optimal actions and their corresponding worst-case disturbances for opinion formation as a result of the optimal control problem in (\ref{eq:dynamics-g}) and (\ref{eq:quadratic-cost-g}) are
 \begin{align}
         u&(t)=-R^{-1}B^\top\Phi^\top(t_f,t)\Big[(\bar{F}+\mathcal{L})\bar{H}^{-1}(t_f)\bar{G}(t_f)-\bar{F}\Big]x_0, \label{eq:optimal}\\
         \varpi&(t)=R_\varpi^{-1}B_\varpi^\top\Phi^\top(t_f,t)\Big[(\bar{F}+\mathcal{L})\bar{H}^{-1}(t_f)\bar{G}(t_f)-\bar{F}\Big]x_0, \label{eq:optimal-dist}
 \end{align}
 where
  \begin{align*}
     &\bar{H}(t_f)=I-\bar{\Psi}(t_f,0)(\bar{F}+\mathcal{L}), \\ &\bar{G}(t_f)=\Phi(t_f,0)-\bar{\Psi}(t_f,0)\bar{F},\\
     &\bar{\Psi}(t_f,0)=\int_0^{t_f}\Phi(t_f,\tau)\big(B_{\varpi}R_{\varpi}^{-1}B_{\varpi}^\top-BR^{-1}B^\top\big)\Phi^\top(t_f,\tau)~\mathrm{d\tau}.
 \end{align*}
Furthermore, the associated state trajectory is given by
\begin{equation}\label{eq:optimal-tr}
     x(t)=\Big(\Phi(t,0)+\bar{\Psi}(t,0)\Big[(\bar{F}+\mathcal{L})\bar{H}^{-1}(t_f)\bar{G}(t_f)-\bar{F}\Big]\Big)x_0.  
\end{equation}
\end{theorem}
\begin{proof}
Applying Pontryagin’s principle, the optimal control problem in (\ref{eq:dynamics-g}) and (\ref{eq:quadratic-cost-g}) converts to the following boundary value problem 
\begin{align*}
    &\dot{x}(t)=-\mathcal{L}x(t)+\big(B_\varpi R_\varpi^{-1}B_\varpi^\top-BR^{-1}B^\top\big)\bar{\lambda}(t), \quad x(0)=x_0,\\
    &\dot{\bar{\lambda}}(t)=\mathcal{L} \bar{\lambda}(t), \quad \bar{\lambda}(t_f)=\big(\bar{F}+\mathcal{L}\big)x(t_f)-\bar{F}x_0.
\end{align*}
Along with that, we have 
\begin{align*}
    &u(t)=-R^{-1}B^\top\bar{\lambda}(t),\\
     &\varpi(t)=R_{\varpi}^{-1}B_{\varpi}^\top\bar{\lambda}(t).
\end{align*}

The solution for $\dot{\bar{\lambda}}(t)$ is $\bar{\lambda}(t)=\Phi(t_f,t)\bar{\lambda}(t_f)$. Substituting this solution in the equation for $\dot{x}(t)$, and solving it for $t_f$, we get
\begin{align*}
x(t_f)=\Phi(t_f,0)x_0+\bar{\Psi}(t_f,0)\bar{\lambda}(t_f).
\end{align*}

Substituting $\bar{\lambda}(t_f)$ into $x(t_f)$ above yields 
 \begin{align*}
x(t_f)=\Phi(t_f,0)x_0+\bar{\Psi}(t_f,0)(\bar{F}+\mathcal{L})x(t_f)-\bar{\Psi}(t_f,0)\bar{F}x_0,
\end{align*}
or,
 \begin{align*}
\Big(I-\bar{\Psi}(t_f,0)(\bar{F}+\mathcal{L})\Big)x(t_f)=\Big(\Phi(t_f,0)-\bar{\Psi}(t_f,0)\bar{F}\Big)x_0.
\end{align*}
Using the notation $\bar{H}(t_f)$ and $\bar{G}(t_f)$, and by appropriate substitutions similar to the proof in Theorem~\ref{theorem:open-Nash}, we obtain (\ref{eq:optimal}) and (\ref{eq:optimal-dist}) as well as (\ref{eq:optimal-tr}).
\end{proof}

Under socially optimal actions, the total terminal opinion formation error is given by 
\begin{align} \label{eq:terminal-error}
E_o&=\frac{1}{2}\sum_{i\in \mathcal{V}}\Vert x_i(t_f) - x_i(0)\Vert_{W_{ii}}^2+\frac{1}{2} \sum_{\{i,j\}\in \mathcal{E}}\Vert x_i(t_f) - x_j(t_f)\Vert_{W_{ij}}^2\nonumber\\&=\frac{1}{2} \Vert x(t_f) - x_0\Vert_{\bar{F}}^2 +\frac{1}{2}\Vert x(t_f)\Vert_{\mathcal{L}}^2.
\end{align}

\section{Numerical Analysis}\label{section:simulation}
For the numerical analysis of the proposed game optimization model, we consider a small social network with five agents (i.e., $n=5$) in a two-dimensional opinion space (i.e., $m=2$). The underlying social graph topology is adapted from~\cite{Ahn}. This example graph was used to analyze opinion formation with the non-stubborn non-strategist agent model on a time-invariant undirected social graph. We consider this graph with directed edges.

The social graph topology is given in~Fig.~\ref{fig:graph} and the initial opinions of agents and the interpersonal influence matrices from~\cite{Ahn} are as $x_1(0)=[1,2]^\top$, $x_2(0)=[2,4]^\top$, $x_3(0)=[3,1]^\top$, $x_4(0)=[4,3]^\top$, and $x_5(0)=[5,6]^\top$, $W_{12}=W_{13}=W_{23}=W_{34}=W_{45}=I_{2\times 2}$. Note that all the weights of the above matrices are considered nonnegative; while zero denotes no coupling, a nonzero weight shows a positive coupling between the associated topics from two agents. This choice of interpersonal influence matrices means that the agents communicate on several unrelated topics on the information graph. As we expected and as seen from Fig.~\ref{fig:no-optimiz}(a), a consensus of final opinions about the average opinion in the network for the non-stubborn non-strategist agent model (\ref{eq:dynamics00}) was reached at a horizon length of $t_f=10$. To examine the evolution of opinions on  interdependent topics, we choose $W_{12}=W_{13}=W_{23}=W_{34}=W_{45}=\begin{bmatrix}
    1&1\\1&1
\end{bmatrix}$, so for all neighboring agents, both topics are fully coupled. Fig.~\ref{fig:no-optimiz}(b) shows that only agents 1 and 5 reach a consensus on both topics.
    
\begin{figure}[ht]
\centering
       \includegraphics[width=0.4\textwidth]{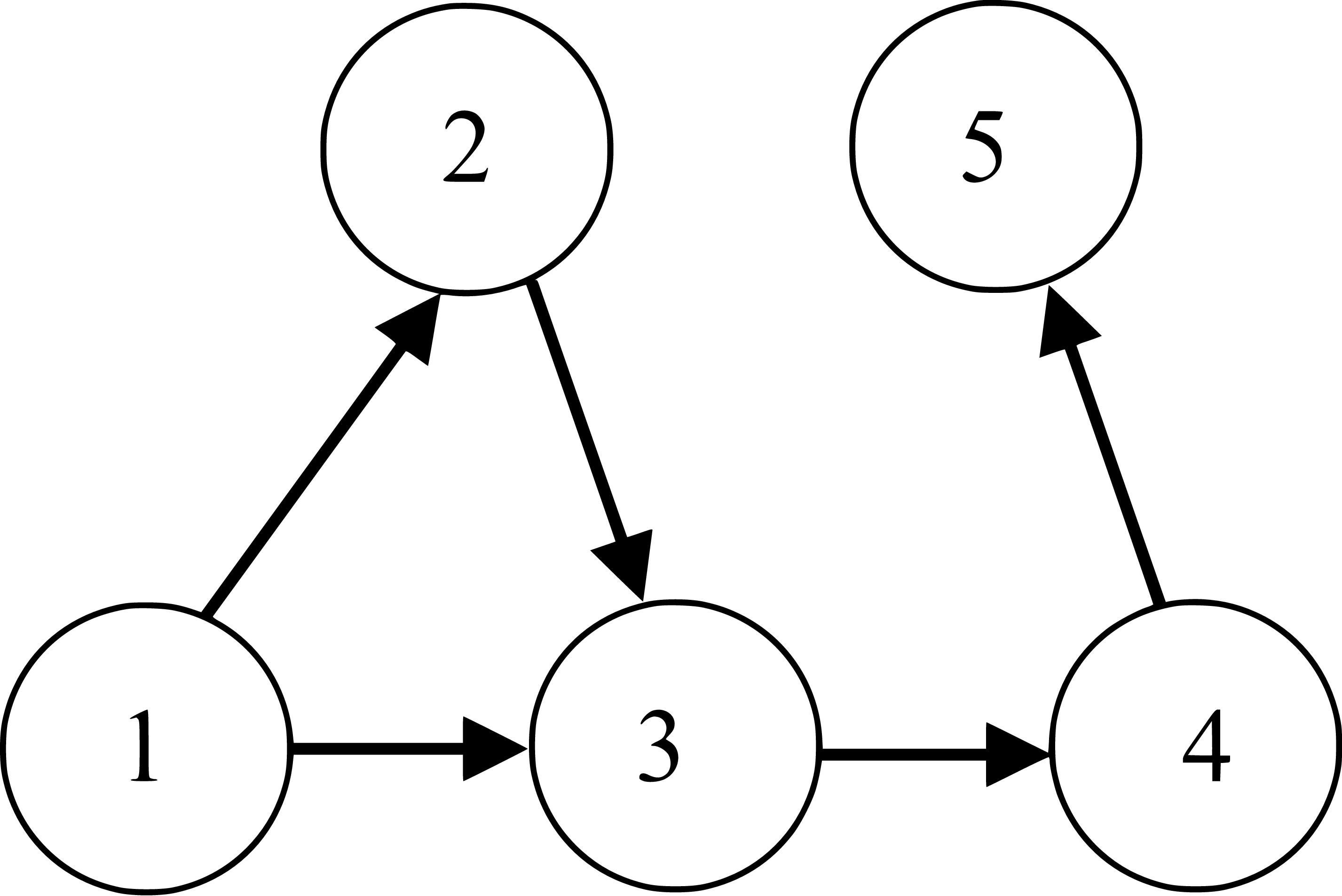}
\caption{Underlying communication graph topology for numerical analysis.}
\label{fig:graph}
\end{figure} 

\begin{figure}[ht]
   \centering
     \begin{subfigure}[b]{0.49\textwidth}
         \centering
         \includegraphics[width=\textwidth]{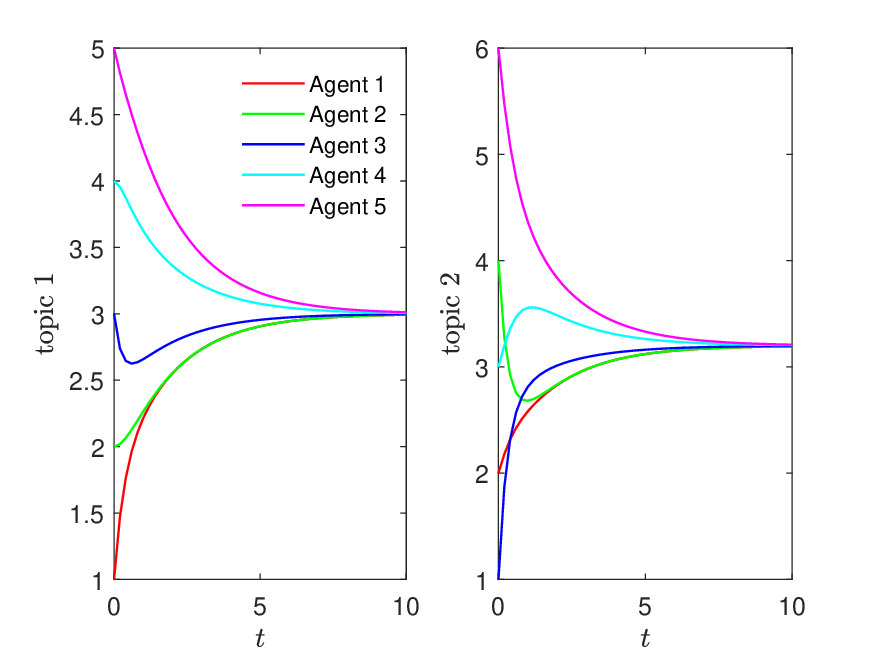}
         \caption{$W_{ij}=I_{2\times2},i=1,\cdots,5,j\in\mathcal{N}_i$}
     \end{subfigure}
     \begin{subfigure}[b]{0.49\textwidth}
         \centering
         \includegraphics[width=\textwidth]{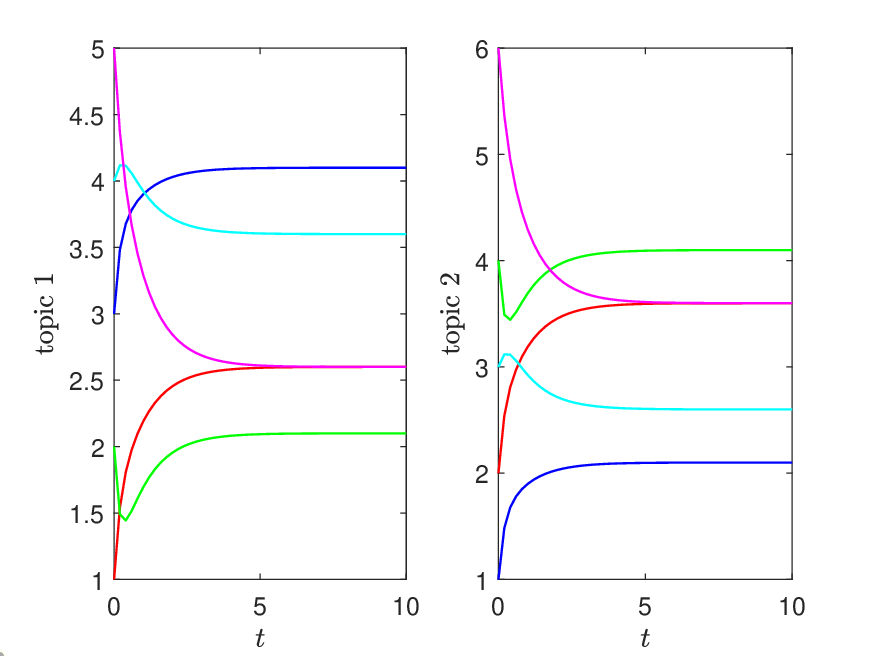}
         \caption{$W_{ij}=[1~1;~1~1],i=1,\cdots,5,j\in\mathcal{N}_i$}
     \end{subfigure}
\caption{Opinion trajectories associated with the non-stubborn non-strategist agent model (\ref{eq:dynamics00}) prior to optimization. (a) A consensus of final opinions about the average opinion in the network was reached. (b) Only agents 1 and 5 reached a consensus.} 
\label{fig:no-optimiz}
\end{figure} 

We continue the numerical analysis with the proposed game optimization model to investigate the impact of differential game optimization on opinion trajectories when the agents are strategists and non-stubborn (i.e., $W_{ii}=0$ for all $i$), and then when they are stubborn (i.e., $W_{ii}=I_{2\times 2}$ for all $i$).
Note that for a time-varying graph, $\mathcal{L}(t)$ satisfies the commutative condition, i.e., $\mathcal{L}(t)\mathcal{L}(\tau) = \mathcal{L}(\tau)\mathcal{L}(t)$ for all $t$ and $\tau$ in $[0,t_f]$. Therefore, the state transition matrix $\Phi(t_f,t)$ is found in explicit form by the spectral decomposition given in (\ref{eq:spec-decom}) in~\ref{app:time-varying}. Let $b_i=\xi_i=I_{2\times 2}$, $R_{i}=r_iI_{2\times 2}$ and $R_{\varpi_i}=r_{\varpi_i}I_{2\times 2}$, ($r_{\varpi_i}\neq r_i$) for $i\in\{1,\cdots,5\}$, then
\begin{align*}
&\Psi_i(t_f,0)=(r_{\varpi_i}-r_i)\int_0^{t_f}\Phi(t_f,t)\Phi^\top(t_f,t)~\mathrm{dt}. 
 \end{align*}

The outcome of the game optimization problem can be quantified by the results presented in Theorem~\ref{theorem:open-Nash} and their distributed counterparts in Section~\ref{section:distributed}. Fig.~\ref{fig:game} shows the evolution of the opinions associated with the game strategies as well as their distributed counterparts for the parameter selection of $r_i=2$, $r_{\varpi_i}=2$, and $\delta_i=1$. Although the distributed estimated opinion trajectories at the beginning are far from the opinion trajectories corresponding to the game strategies, they converge near the middle point of the horizon length. 

Let all agents expect significant disturbances and thereby choose a smaller $r_{\varpi_i}$. In Fig.~\ref{fig:game-dist}, we re-illustrate the results for $r_{\varpi_i}=0.5$. Compared to Fig.~\ref{fig:game} where $r_{\varpi_i}=2$, we observe that only for stubborn agents with unrelated topics ($W_{ii}=I_{2\times2},W_{ij}=I_{2\times2}$), their opinion trajectories change. Therefore, the decisions of stubborn agents on unrelated topics are affected the most by the changes in disturbance expectations.

\begin{figure}[ht]
\centering
     \begin{subfigure}[b]{0.49\textwidth}
         \centering
         \includegraphics[width=\textwidth]{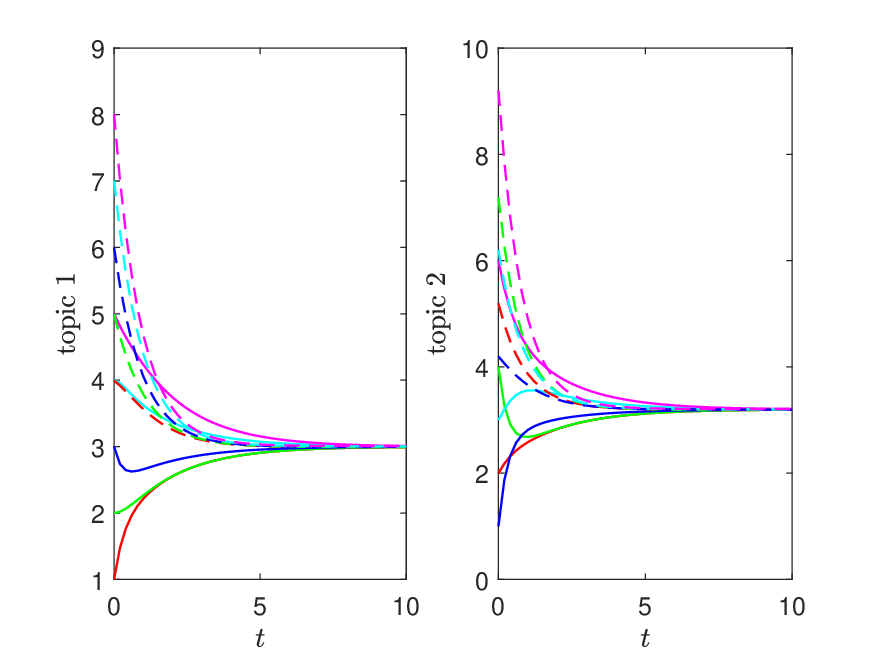}
         \caption{$W_{ii}=0_{2\times2},W_{ij}=I_{2\times2},r_{\varpi_i}=2$}
     \end{subfigure}
     \begin{subfigure}[b]{0.49\textwidth}
         \centering
         \includegraphics[width=\textwidth]{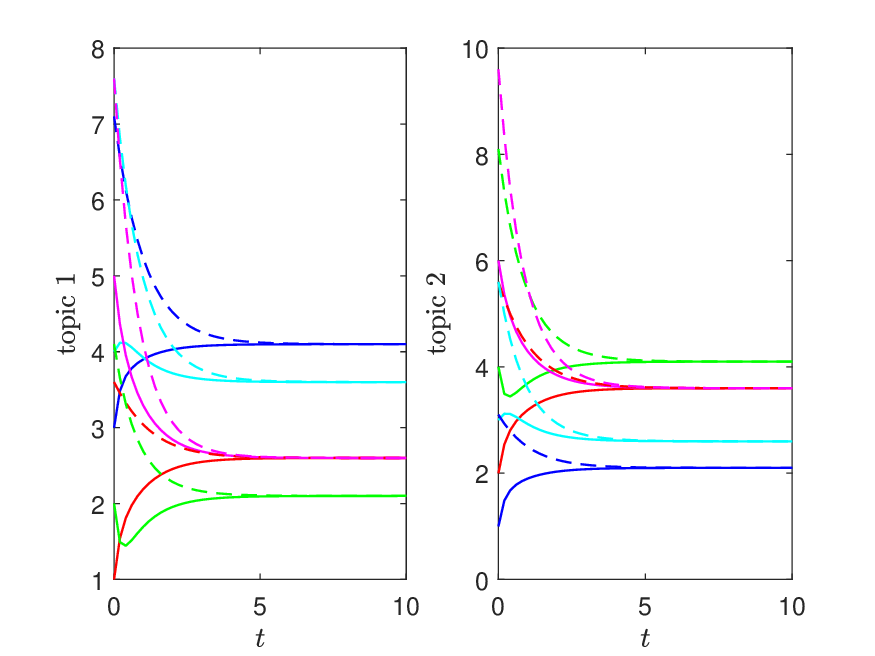}
         \caption{$W_{ii}=0_{2\times2},W_{ij}=[1~1;~1~1],r_{\varpi_i}=2$}
     \end{subfigure}
    \begin{subfigure}[b]{0.49\textwidth}
         \centering
         \includegraphics[width=\textwidth]{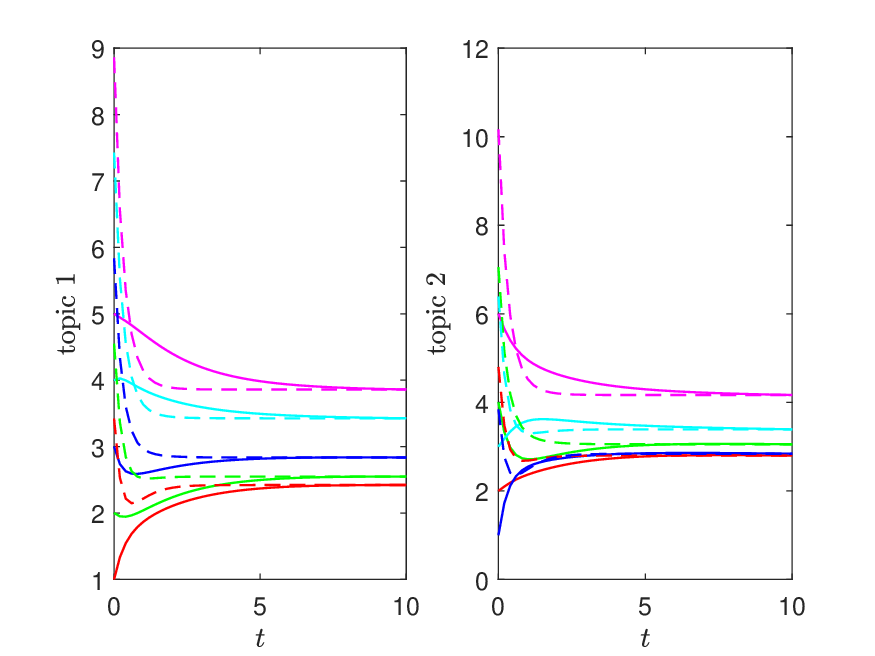}
         \caption{$W_{ii}=I_{2\times2},W_{ij}=I_{2\times2},r_{\varpi_i}=2$}
     \end{subfigure}
     \begin{subfigure}[b]{0.49\textwidth}
         \centering
         \includegraphics[width=\textwidth]{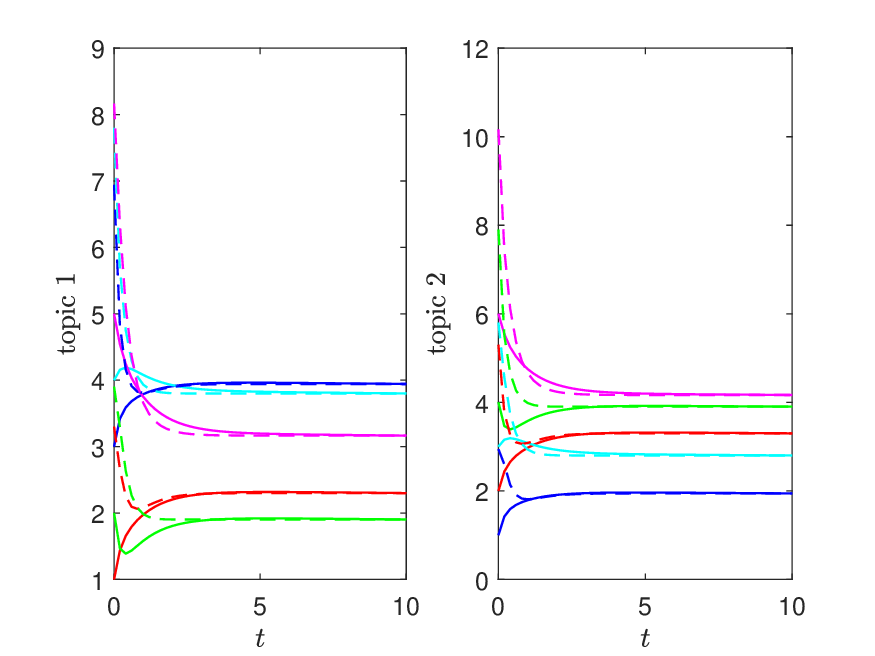}
         \caption{$W_{ii}=I_{2\times2},W_{ij}=[1~1;~1~1],r_{\varpi_i}=2$}
     \end{subfigure}
\caption{Opinion trajectories associated with the game strategies and their distributed counterparts (shown with dashed lines). (a) non-stubborn agents with unrelated topics; (b) non-stubborn agents with fully coupled topics; (c) stubborn agents with unrelated topics; (d) stubborn agents with fully coupled topics.}
\label{fig:game}
\end{figure}

\begin{figure}[ht]
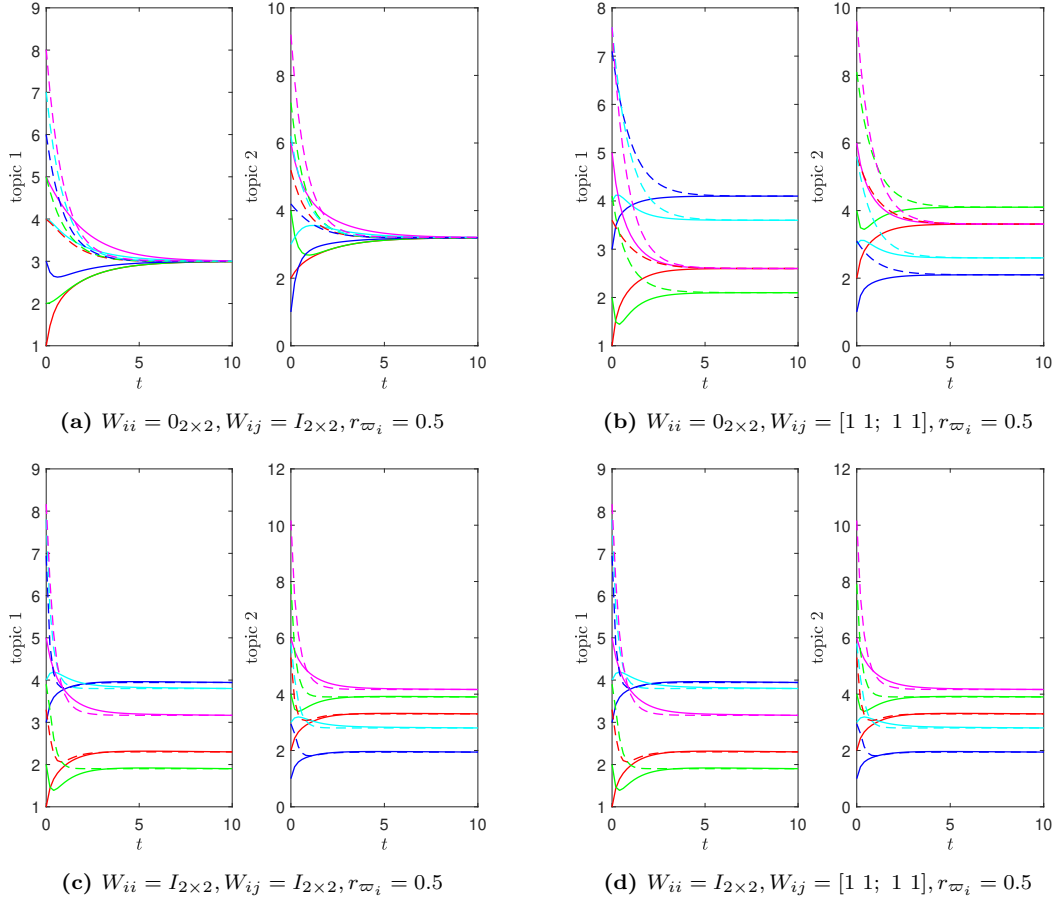

\centering
     \begin{subfigure}[b]{0.49\textwidth}
         \centering
         \includegraphics[width=\textwidth]{WiiZeroWijEyeRw05.eps}
         \caption{$W_{ii}=0_{2\times2},W_{ij}=I_{2\times2},r_{\varpi_i}=0.5$}
     \end{subfigure}
     \begin{subfigure}[b]{0.49\textwidth}
         \centering
         \includegraphics[width=\textwidth]{WiiZeroWijOnesRw05.eps}
         \caption{$W_{ii}=0_{2\times2},W_{ij}=[1~1;~1~1],r_{\varpi_i}=0.5$}
     \end{subfigure}
    \begin{subfigure}[b]{0.49\textwidth}
         \centering
         \includegraphics[width=\textwidth]{WiiEyeWijEyeRw05.eps}
         \caption{$W_{ii}=I_{2\times2},W_{ij}=I_{2\times2},r_{\varpi_i}=0.5$}
     \end{subfigure}
     \begin{subfigure}[b]{0.49\textwidth}
         \centering
         \includegraphics[width=\textwidth]{WiiEyeWijOnesRw05.eps}
         \caption{$W_{ii}=I_{2\times2},W_{ij}=[1~1;~1~1],r_{\varpi_i}=0.5$}
     \end{subfigure}
\caption{Opinion trajectories associated with the game strategies and their distributed counterparts with disturbances $r_{\varpi_i}=0.5$ for $i=1,\cdots,5$.}
\label{fig:game-dist}
\end{figure}

The results of applying the socially optimal actions presented in Theorem~\ref{theorem:optimal} are shown in Fig.~\ref{fig:optimal}. For non-stubborn agents, their opinion trajectories under the game strategies and social optimal actions overlap, while for stubborn agents, they are similar. Agent 5 has the least communication with other agents, and only its final opinion slightly differs when under game strategies and socially optimal actions. 

\begin{figure}[ht]
\centering
     \begin{subfigure}[b]{0.49\textwidth}
         \centering
         \includegraphics[width=\textwidth]{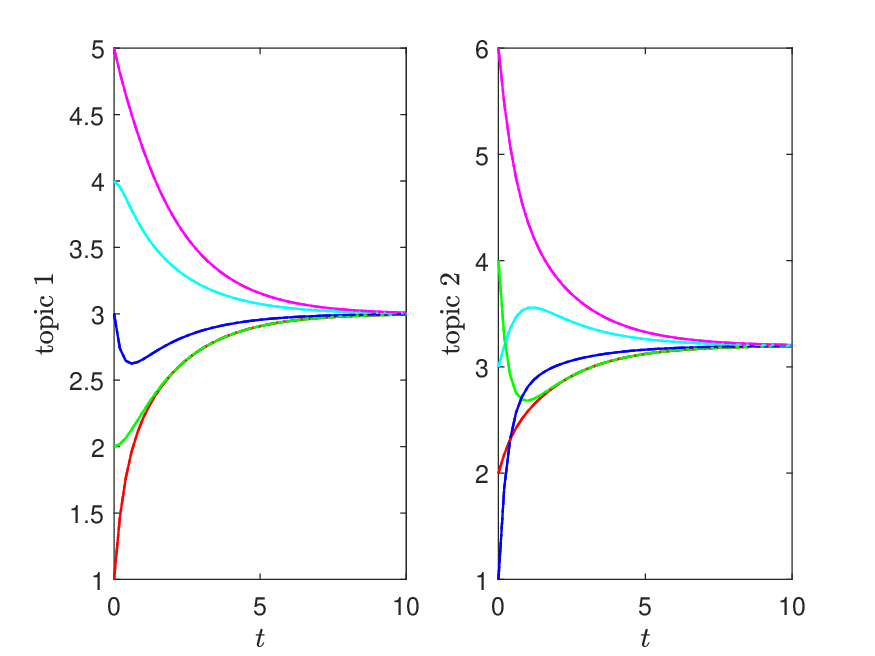}
         \caption{$W_{ii}=0_{2\times2},W_{ij}=I_{2\times2},r_{\varpi_i}=0.1$}
     \end{subfigure}
     \begin{subfigure}[b]{0.49\textwidth}
         \centering
         \includegraphics[width=\textwidth]{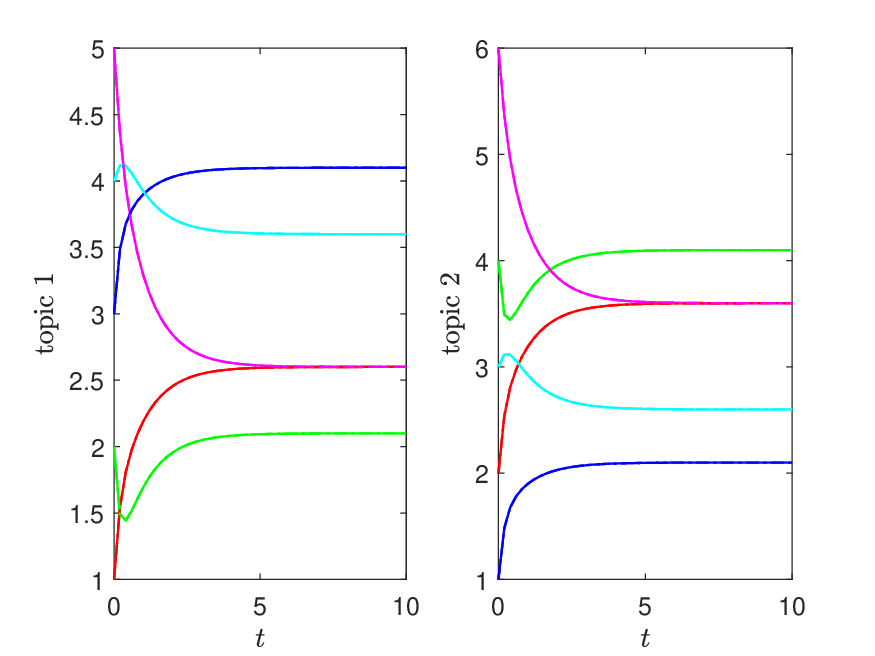}
         \caption{$W_{ii}=0_{2\times2},W_{ij}=[1~1;~1~1],r_{\varpi_i}=0.1$}
     \end{subfigure}
    \begin{subfigure}[b]{0.49\textwidth}
         \centering
         \includegraphics[width=\textwidth]{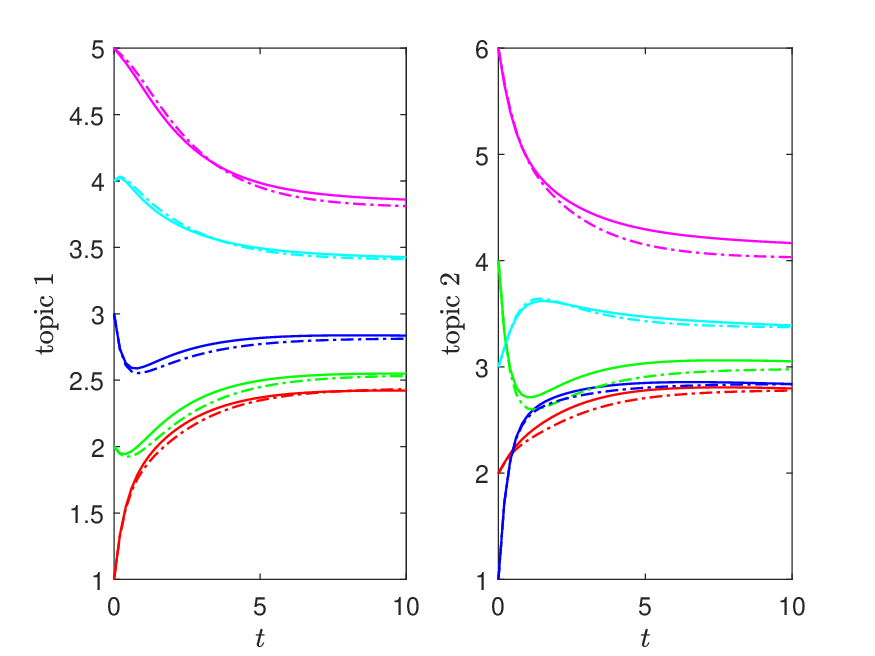}
         \caption{$W_{ii}=I_{2\times2},W_{ij}=I_{2\times2},r_{\varpi_i}=0.1$}
     \end{subfigure}
     \begin{subfigure}[b]{0.49\textwidth}
         \centering
         \includegraphics[width=\textwidth]{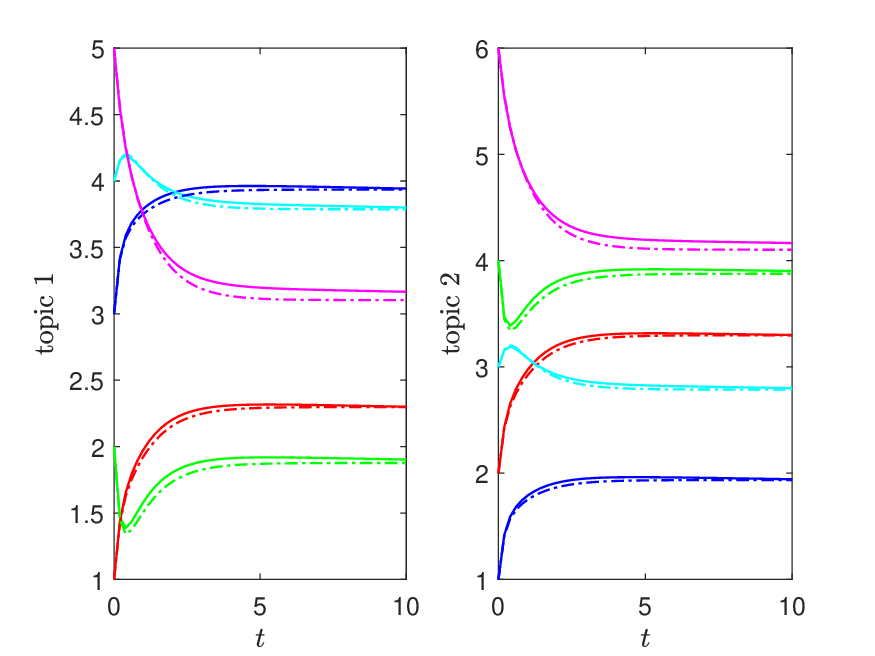}
         \caption{$W_{ii}=I_{2\times2},W_{ij}=[1~1;~1~1],r_{\varpi_i}=0.1$}
     \end{subfigure}
\caption{Opinion trajectories associated with socially optimal actions (shown with a dash-dotted line) versus opinion trajectories associated with game strategies.}
\label{fig:optimal}
\end{figure} 

To see whether the changes in the stubbornness of stubborn agents affect their opinion trajectories associated with game strategies and socially optimal actions, we show both trajectories for highly stubborn agents with the stubbornness of $W_{ii}=10I_{2\times 2}$ for $i=1,\cdots,5$ in Figure~\ref{fig:stubborn}. We observe that at the beginning and end of the horizon length, both trajectories are very close, but the rest of the time they get a little far from each other. 

\begin{figure}[ht]
\centering
     \begin{subfigure}[b]{0.49\textwidth}
         \centering
         \includegraphics[width=\textwidth]{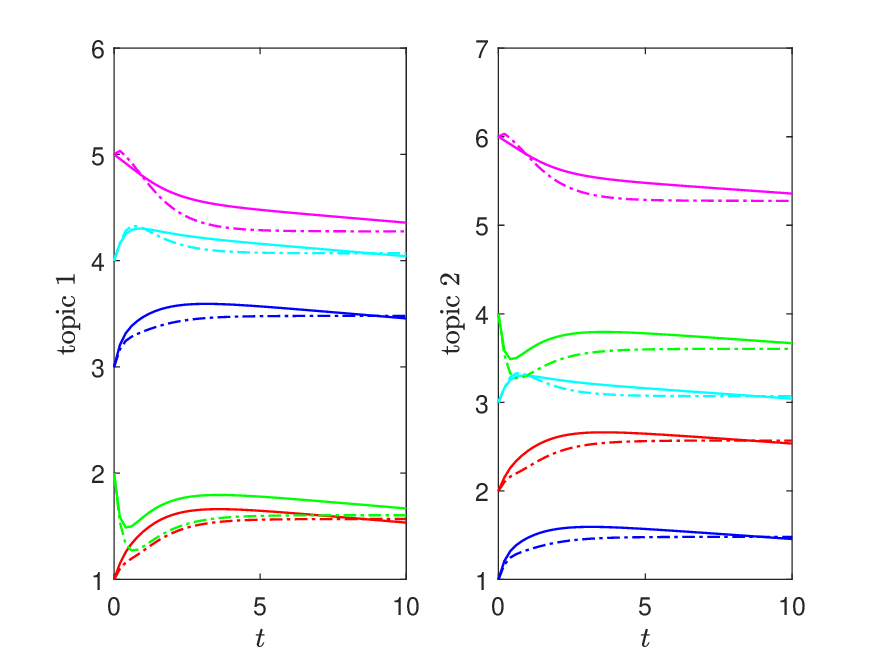}
         \caption{$W_{ii}=10I_{2\times2},W_{ij}=I_{2\times2},r_{\varpi_i}=0.1$}
         \label{fig:sim1}
     \end{subfigure}
     \begin{subfigure}[b]{0.49\textwidth}
         \centering
         \includegraphics[width=\textwidth]{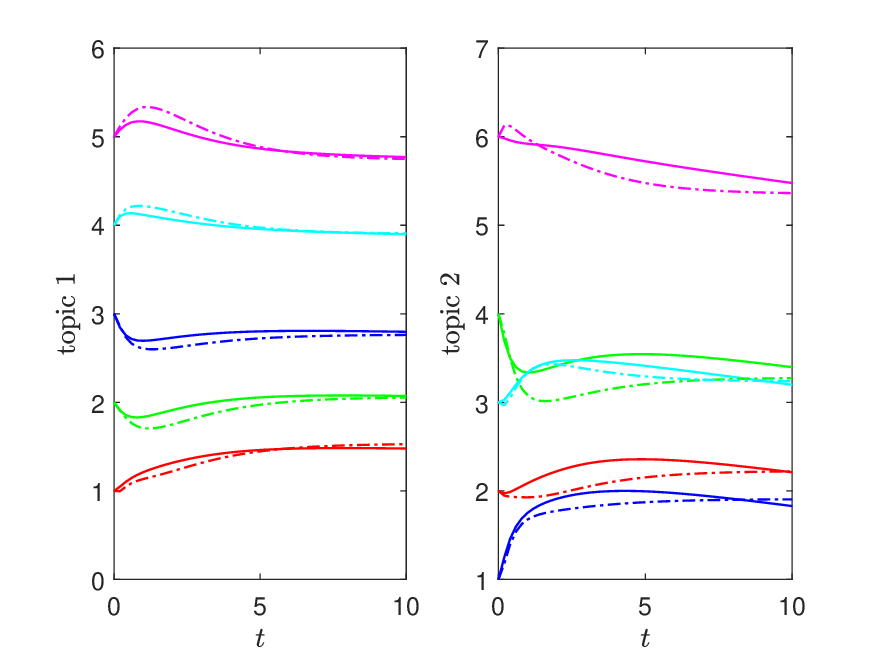}
         \caption{$W_{ii}=10I_{2\times2},W_{ij}=[1~1;~1~1],r_{\varpi_i}=0.1$}
     \end{subfigure}
\caption{Opinion trajectories associated with game strategies and socially optimal actions for highly stubborn agents.}
\label{fig:stubborn}
\end{figure} 

Furthermore, to quantify the performance of the game strategies against their distributed counterparts and the socially optimal actions, we utilize the terminal opinion formation errors that correspond to each. The terminal opinion formation error for agent $i$ under the game strategies and their distributed counterparts are calculated from (\ref{eq:terminal-error-i}) and (\ref{eq:terminal-error-i-dist}), respectively. Under socially optimal actions, the total terminal opinion formation error is given by (\ref{eq:terminal-error}). The terminal opinion formation errors under the game strategies ($E_i$), their distributed counterparts ($\hat{E}_i$), and the socially optimal actions ($E_o$) for $t_f=10$ are summarized in Tables~\ref{table:WiiZeroWijEye}-\ref{table:WiiEyeWijOnes}. From these tables, we see that the errors under the game strategies and the errors under their distributed counterparts are very close to each other, and they are less than the errors under the socially optimal actions. In addition, the errors for stubborn agents are much higher than for non-stubborn agents. Although the socially optimal action approach has nearly the same performance as the game strategy approach and its distributed counterpart approach in terms of terminal opinion formation errors, it has significant drawbacks that make it undesirable. As it is apparent from (\ref{eq:optimal}), each agent has to acquire the complete information of all other agents and their initial state vector. On the contrary, the agents can implement the counterparts of game strategies distributedly. In addition, social optimality does not allow individual agents to be strategists, who are decision-makers prioritizing their own interests.

\begin{table}[ht]
\small
\centering
\caption{Terminal opinion formation errors for non-stubborn ($W_{ii}=0_{2\times 2}$) agents on unrelated topics ($W_{ij}=I_{2\times 2}$)} 
\begin{tabular}{|p{0.5cm}|p{2.5cm}|p{2.5cm}|p{2.5cm}|}
 \hline
 $i$  & $E_i$  & $\hat{E}_i$ & $E_o$ 
 \\\hline
1&  $7.4116\times 10^{-6}$ & $7.4229\times 10^{-6}$ & $5.5011\times 10^{-5}$ \\
2&  $7.4116\times 10^{-6}$ & $7.4206\times 10^{-6}$ & $5.5011\times 10^{-5}$\\ 
3&  $5.9987\times 10^{-5}$ & $6.0077\times 10^{-5}$ & $5.5011\times 10^{-5}$ \\
4&  $6.5443\times 10^{-5}$ & $6.5548\times 10^{-5}$ & $5.5011\times 10^{-5}$\\
5&  $2.0279\times 10^{-5}$ & $2.0314\times 10^{-5}$ & $5.5011\times 10^{-5}$ \\
\hline
\end{tabular}
\label{table:WiiZeroWijEye}
\end{table}

\begin{table}[ht]
\small
\centering
\caption{Terminal opinion formation errors non-stubborn agents ($W_{ii}=0_{2\times 2}$) on fully coupled topics ($W_{ij}=[1~1;~1~1]$)} 
\begin{tabular}{|p{0.5cm}|p{2.5cm}|p{2.5cm}|p{2.5cm}|}
 \hline
 $i$  & $E_i$  & $\hat{E}_i$ & $E_o$ 
 \\\hline
1&  $4.5838\times 10^{-10}$ & $5.8660\times 10^{-10}$ & $3.4023\times 10^{-9}$ \\
2&  $4.5839\times 10^{-10}$ & $5.5744\times 10^{-10}$ & $3.4023\times 10^{-9}$\\ 
3&  $3.7100\times 10^{-9}$ & $4.7495\times 10^{-9}$ & $3.4023\times 10^{-9}$\\
4&  $4.0475\times 10^{-9}$ & $5.2724\times 10^{-9}$ & $3.4023\times 10^{-9}$\\
5&  $1.2542\times 10^{-9}$ & $1.6661\times 10^{-9}$ & $3.4023\times 10^{-9}$ \\
\hline
\end{tabular}
\label{table:WiiZeroWijOnes}
\end{table}
\begin{table}[ht]
\small
\centering
\caption{Terminal opinion formation errors for stubborn ($W_{ii}=I_{2\times 2}$) agents on unrelated topics ($W_{ij}=I_{2\times 2}$)} 
\begin{tabular}{|p{0.5cm}|p{1.5cm}|p{1.5cm}|p{1.5cm}|}
 \hline
 $i$ & $E_i$ & $\hat{E}_i$ & $E_o$ 
 \\\hline
1&  0.2602 & 0.2602 & 1.4263 \\
2&  0.1645 & 0.1645 & 1.4263 \\ 
3&  0.6474 & 0.6474 & 1.4263 \\
4&  0.7448 & 0.7448 & 1.4263 \\
5&  0.6275 & 0.6275 & 1.4263 \\
\hline
\end{tabular}
\label{table:WiiEyeWijEye}
\end{table}
\begin{table}[ht]
\small
\centering
\caption{Terminal opinion formation errors for stubborn ($W_{ii}=I_{2\times 2}$) agents on fully coupled topics ($W_{ij}=[1~1;~1~1]$)} 
\begin{tabular}{|p{0.5cm}|p{1.5cm}|p{1.5cm}|p{1.5cm}|}
 \hline
 $i$  & $E_i$  & $\hat{E}_i$ & $E_o$ 
 \\\hline
1&  0.2307 & 0.2307 & 1.1261 \\
2&  0.0251 & 0.0251 & 1.1261\\ 
3&  0.3882 & 0.3882 & 1.1261 \\
4&  0.5265 & 0.5265 & 1.1261\\
5&  0.6037 & 0.6037 & 1.1261 \\
\hline
\end{tabular}
\label{table:WiiEyeWijOnes}
\end{table}

\section{Conclusions}\label{section:conclusions}

In this study, a differential game model of opinion formation in social networks has been proposed. The evolution of the opinions of self-interested people in social groups formed through friendships in real life served as the model's inspiration. The game problem has been solved for an open-loop information Nash equilibrium solution. Moreover, a distributed implementation of the solution has been proposed where the game strategies and the associated opinion trajectory for each player solely depend on the information from the neighboring players. The model was used to examine the opinion formation of a small opinion network with five agents. A natural extension to the proposed model can be a time-varying graph topology with antagonistic agent interactions.

\section*{Acknowledgments}

This work was supported by the Czech Science Foundation (GACR) grant no. 23-07517S. This work was supported in full by the Science and Research Council of Turkey (T\"{U}B\.{I}TAK) under project EEEAG-121E162.

 \appendix
\section{Commutative linear time-varying systems} \label{app:time-varying}

Consider the linear time-varying differential state space systems described by 
\begin{equation}\label{eq:STM}
    \dot{x}(t)=A(t)x(t), \quad x(t_0)=x_{0},
\end{equation}
where $x(t)$ is the state vector, $A(t)$ is a time-varying matrix, and $x_0$ is an initial condition. The solution of (\ref{eq:STM}) is given by 
\begin{equation*}
    x(t)=\Phi(t,t_0)x_{0},
\end{equation*}
where $\Phi(t,t_0)$ denotes the state-transition matrix of $A(t)$ and satisfies
\begin{equation*}
    \dot{\Phi}(t,t_0)=A(t)\Phi(t,t_0), \quad \Phi(t_0,t_0)=I.
\end{equation*}

 In general, the state-transition matrix for time-varying systems is not available in explicit form. 
 For a special class of linear time-varying systems above, $A(t)$ commute with $A(\tau)$ for all $t$ and $\tau$. For these systems, the unique state-transition matrix is explicitly available in the exponential form given below
\begin{equation*}
    \Phi(t,t_0)=\exp{(\int_{t_0}^t A(\tau)d\tau)}.
\end{equation*}
Under the similarity transformation, the expression for the state-transition matrix further simplifies as
\begin{equation*}
    \Phi(t,t_0)=P(t)\exp{(\int_{t_0}^t \Lambda(\tau)d\tau)}P^{-1}(t),
\end{equation*}
where $\Lambda(t)$ is a diagonal matrix and $P(t)$ is a nonsingular modal matrix $P(t)$~\cite{commutative}. 

A more effective way of calculation of $\Phi(t,t_0)$ is using the concept of an extended eigenvalue and extended eigenvector introduced in~\cite{1102448}. Let the scalar $\mu(t)$ be an extended eigenvalue and $e(t)$ be its corresponding extended eigenvector and the pair $\{\mu(t),e(t)\}$ is called the extended eigenpair of $A(t)$. Then, the spectral decomposition of the state-transition matrix is explicitly given by
\begin{equation}\label{eq:spec-decom}
    \Phi(t,t_0)=\sum_{i}\exp{\Big(\int_{t_0}^t \mu_i(\tau)d\tau\Big)}e_i(t)r_i^\top(t_0),
\end{equation}
where $r_i^\top(t)$ is the reciprocal basis of $e(t)$, i.e.,
\begin{equation*}
    r_i^\top(t)e_j(t)=\begin{cases}
      0 & \text{$i\neq j$},\\
      1 & \text{$i= j$}.
    \end{cases}       
\end{equation*}
For a commutative $A(t)$ its extended eigenvectors are constant~\cite{Wang2017} and its extended eigenpair are determined by the auxiliary equations 
\begin{align*}
    \mathrm{det}\big(\mu(t)I-A(t)\big)=0, \quad \big(\mu(t)I-A(t)\big)e=0.
\end{align*}

When $A$ is time-invariant and all of its eigenvalues have negative real parts, (\ref{eq:STM}) becomes asymptotically stable. For a time-varying $A(t)$, (\ref{eq:STM}) is said to be asymptotically stable if and only if $x(t)$ starting from any finite initial state $x_0$ is bounded and $x(t)\rightarrow 0$ as $t\rightarrow \infty$. The asymptotic stability of (\ref{eq:STM}) can be characterized in terms of its extended eigenpairs as quoted below from Theorem 3 in~\cite{Wang2017}. Define 
\begin{equation*}
    \Omega_i=\Vert \exp{\Big(\int_{t_0}^t \mu_i(\tau)d\tau\Big)}e_i(t)\Vert, 
\end{equation*}
for each extended eigenpair $\{\mu(t),e(t)\}$ of $A(t)$. If for every $\{\mu_i(t),e_i(t)\}$, $\Omega_i\leq \Xi$ for all $t\geq t_0$ and $\Xi<\infty$, then $x(t)$ is bounded, and  $\Omega_i\rightarrow 0$ as $t\rightarrow \infty$, the time varying system (\ref{eq:STM}) is asymptotically stable.

\end{document}